%% file: fossacs_ff.tex
\documentclass[runningheads]{llncs}
\input{macros}

\usepackage[T1]{fontenc}
\usepackage{tikz}
\usepackage{tikz-cd}
\usepackage{blkarray}
\usepackage{pgfplots}
\pgfplotsset{compat=1.18}
\usepackage{graphicx}
\usepackage{comment}
\usepackage{leftindex}
\usepackage{amssymb}
\usepackage{nicefrac}
\usepackage{xcolor}

\usepackage{MnSymbol}
\usepackage{algorithm}
\usepackage{algpseudocode}
\usepackage{adjustbox}
\usepackage{relsize}
\usepackage{enumerate}

\newcommand{\NFTodo}[1]{\noindent {\color{teal}{Natasha TODO: #1}}}

\makeatletter
\renewcommand*\env@matrix[1][*\c@MaxMatrixCols c]{%
  \hskip -\arraycolsep
  \let\@ifnextchar\new@ifnextchar
  \array{#1}}
\makeatother

\begin{document}
\title{Composition Theorems for $f$-Differential Privacy}

\author{Natasha Fernandes\inst{1}
\and
Annabelle McIver\inst{1} \and
Parastoo Sadeghi\inst{2}}
%

\institute{Macquarie University, Sydney \\\email{natasha.fernandes@mq.edu.au, annabelle.mciver@mq.edu.au}
 \and
UNSW, Canberra\\
\email{p.sadeghi@unsw.edu.au}
}

\maketitle              
\begin{abstract}
$f$-differential privacy ($f$-DP) is a recent definition for privacy which can offer improved predictions of  ``privacy loss''. It has been used to analyse specific privacy mechanisms, such as the popular Gaussian mechanism.

In this paper we show how $f$-DP's foundation in statistical hypothesis testing implies equivalence to the channel model of Quantitative Information Flow (QIF). We demonstrate this equivalence as a Galois connection between two partially-ordered sets, namely $f$-DP's trade-off functions, and a class of information channels. This equivalence enables  novel general composition theorems for $f$-DP, supporting improved analysis for complex privacy designs. We apply our results to the popular privacy amplification mechanisms of sub-sampling and purification, to produce novel $f$-DP profiles for these general privacy-enhancing algorithms. 

\end{abstract}
\section{Introduction}



Since the first definitions for privacy were introduced \cite{10.1142/S0218488502001648,Dwork:2006aa,8049725} the principles underlying privacy protections  have become  steadily more refined. It is now generally recognised that a privacy definition should be based on a measurement of ``information leakage'', however defined, and that the measurement should satisfy some version of the ``data processing inequality'',  namely that post-processing after a data release should only improve privacy.
A second important property of a privacy definition is that it should admit good composition laws, because most implementations of privacy algorithms usually comprise a composition of several different privacy-enhancing schemes, as illustrated in \Alg{Purify-1511} (below) \cite{lin2025purifyingapproximatediferentialprivacy}.


The purpose of \Alg{Purify-1511} is to ``boost'' the privacy properties of an input mechanism $M$ applied to a secret value $x$. This requires a fine-grained analysis of the steps in the computation so that the  information leaks due to $M$ on its own can be compared to the overall information leaks when used in conjunction with  \Alg{Purify-1511}. The challenge here is that the only information about $M$  is the privacy definition it satisfies. The traditional approach to verifying such implementations is to develop techniques based on refinement and abstraction, so that abstractions of program components can be analysed efficiently, with the resulting analysis also applicable to more detailed implementations via refinement.


\begin{algorithm}
\begin{algorithmic}\caption{Privacy purification}\label{Purify-1511}
\Require Mechanism $M:{\cal X}\rightarrow{\cal Y}$ satisfying $(\epsilon, \delta)$-differential privacy, Private input $x$, Parameters: $r\in [0,1]$, $\epsilon'{>0}$
\Ensure  Satisfies a ``pure differential privacy constraint''. (See comment after (\ref{dp-1053}).) 
\State $v\gets U[0,1]$; \Comment{Choose a value uniformly from $[0,1]$}
\If {$(v< r)$} \Comment{Hidden probabilistic choice with $r$ bias}
\State  $y\gets M(x)$ 
\Else
\State $y \gets U[{\cal Y}']$ \Comment{Choose a value uniformly from ${\cal Y}'$}
\EndIf
\State $z\gets y + G_{\epsilon'}(0)$; \Comment{Add Geometric perturbation using parameter $\epsilon'$}
\State Output $z$; \Comment{Output sanitised result}
\end{algorithmic}
\end{algorithm}

In this paper we investigate how to do that for privacy properties using a recent notion called \emph{$f$-differential privacy} \cite{10.1111/rssb.12454} ($f$-DP). Similar to the more established $(\epsilon, \delta)$-differential privacy (or $(\epsilon, \delta)$-DP for short), which is based on ``indistinguishability'' between related scenarios, $f$-DP is based on statistical hypothesis testing, which turns out to support more nuanced evaluations of privacy risks than does traditional DP. 
Although $f$-DP has the potential to provide more accurate  privacy assessments, a significant drawback is that it does not appear to be associated with simple composition laws, making it difficult to use in practice, except for specific mechanisms.

Our goal is to show how to enable accurate, general analysis of algorithms wrt.\ $f$-DP by establishing an equivalence between $f$-DP and information channels, using the theory of Quantitative Information Flow (QIF). Here, QIF provides an extensive theory for fine-grained analysis of information flow in programs \cite{Alvim:20a}, and is therefore suited to modelling combinations of general privacy-enhancing schemes. Moreover it supports analysis at different levels of abstraction through its information refinement order.

In this paper we provide a detailed foundational analysis of $f$-DP in terms of its information leakage properties via QIF. With that understanding, we show how $f$-DP admits a number of universal composition laws, supporting  detailed analysis of privacy-preserving algorithms. {\bf Our contributions} are:

\begin{enumerate}
\item We establish a Galois connection (\Thm{GC-1651}) between two partially ordered sets: the set of $f$-DP's trade-off functions under pointwise less-than $({\mathbb F}, \leq)$, and  the set of QIF's two-row information channels ordered by refinement, $({\mathbb C}_2, \sqsubseteq)$.  The Galois connection is defined by two order-preserving mappings ${\cal T}$ and ${\cal C}$ in \Fig{Contributions}. We discover (\Thm{QIF+HT-1213}) a novel relationship between the leakage measurements in QIF and so-called ``hockey-stick divergence'', as indicated by $({\mathbb C}_2, \leq_{\underline{h}})$ in \Fig{Contributions}. 
\item We study (\Sec{comp-defs}) the behaviour of $f$-DP wrt.\ a range of probabilistic program constructors via their interpretation in ${\mathbb C}_2$, including hidden and visible probabilistic choices, probabilistic perturbation and pre-processing, producing a number of new universal composition laws (\Sec{Univ-comps-1028}). 
\item We demonstrate (\Sec{Applications-main}) our techniques on some popular algorithms, including purification (\Alg{Purify-1511}) and sub-sampling (\Alg{Subsample-poisson-2124}).
\end{enumerate}

\vspace{-12pt}
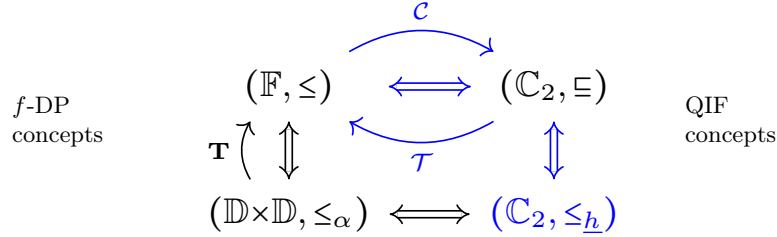
\begin{figure}[!th]
\centering
\begin{minipage}{0.2\textwidth}
$f$-DP \\concepts
\end{minipage}\hspace{12pt}
\begin{minipage}{0.4\textwidth}
\adjustbox{scale=1.5,center}{
\begin{tikzcd}[%
   baseline=0
  ,outer ysep=1px
  ,outer xsep=1px
   ]
    |[minimum width=1.7cm]|  (\mathbb{F}, \leq)
        \arrow[r, bend left, "\mathcal{C}", color=blue]
        \arrow[d, Leftrightarrow] 
        \arrow[r, Leftrightarrow, color=blue]
    & |[minimum width=1.4cm]|(\mathbb{C}_2, \sqsubseteq)
        \arrow[l, bend left, "\cal T", color=blue] 
        \arrow[d, Leftrightarrow, color=blue]  \\
    %
    (\mathbb{D}{\times}\mathbb{D}, \leq_\alpha) 
        \arrow[r, Leftrightarrow] 
        \arrow[u, bend left, shift left=1.2ex, "\textbf{T}"]
    & |[color=blue]|(\mathbb{C}_2, \leq_{\underline{h}}) 
\end{tikzcd}
}
\end{minipage} \hspace{25pt}
\begin{minipage}{0.2\textwidth}
QIF \\concepts
\end{minipage}
\caption{A summary of the relationships between hypothesis testing and quantitative information flow. Our contributions in this paper are highlighted in \textcolor{blue}{blue}.}\label{Contributions}
\end{figure}

\section{Preliminaries: privacy and information flow}




\subsection{Standard Differential Privacy}
A differential privacy mechanism is one that sanitises results of queries to datasets by adding a random perturbation before publication.   We model such mechanisms ${\cal M}$ as a function $\mathcal{D} \to \Dist{\cal Y}$, where $\mathcal{D}$  denotes the set of datasets, and ${\mathbb D}{\cal Y}$ is the set of probability distributions over set ${\cal Y}$. (For simplicity we work under the assumption that ${\cal Y}$ is discrete, showing how to remove it in the appendix.)
A dataset itself consists of a ``set of records''; standard \emph{differential privacy} enforces a constraint on ${\cal M}$ wrt.\ non-negative parameters $(\epsilon, \delta)$, and pairs of ``adjacent datasets''. We say that $D,D'\in {\cal D}$ are \emph{adjacent} ``if they differ by a single record'', which we define here as: $D \subseteq D'$ and $|D'- D|=1$, or vice versa.
Mechanism ${\cal M}$ is said to satisfy $(\epsilon, \delta)$-privacy if, for any $Y\subseteq {\cal Y}$,  and any pair of \emph{adjacent} datasets:
%
\begin{equation}\label{dp-1053}
{\cal M}(D)(Y) ~~\leq~~ e^\epsilon{\cal M}(D')(Y) + \delta~.
\end{equation}
We say that ${\cal M}$ satisfies a {\bf ``pure'' differential privacy} definition if $\delta{=}0$.

We typically visualise this definition as follows. Let ${\cal Y}= \{y_0,\dots, y_n\}$, and  $D_0, D_1$ be a pair of adjacent datasets. We write $p_i={\cal M}(D_0)(\{ y_i\})$ (and $q_i={\cal M}(D_1)(\{ y_i\})$) for the probability that the output $y_i$ is observed when the input to ${\cal M}$ is $D_0$ (or $D_1$). This yields the following ``channel'', $M$ (see \Sec{QIFSec} below):

\begin{equation}\label{M-chan}
M ~~=~~ \begin{bmatrix}
p_0 & p_1 & \dots & p_{n-1} & p_n\\
q_0 & q_1 & \dots & q_{n-1} & q_n
\end{bmatrix}~, ~~~\textit{where}~~~ \sum_i p_i = \sum_i q_i =1~.
\end{equation}

Wlog, assume $M_{D_0}$ is the first row and  $M_{D_1}$ is the second row of this channel. Then, $M$ satisfies $(\epsilon, \delta)$-DP if for all $S\subset \{0,\cdots, n\}$, $p_S \leq e^\epsilon q_S {+}\delta$, and $q_S \leq e^\epsilon p_S {+}\delta$, where $p_S = \sum_{i\in S} p_i$, $q_S = \sum_{i\in S} q_i$.
Not all perturbation methods satisfy $(\epsilon, \delta)$-DP for given parameters. 
When $\epsilon$ and $\delta$ are small, then the probability pairs $(p_i, q_i)$ are more similar to each other, and so it makes it harder to distinguish between the $(D_0, D_1)$ inputs for any observed output. The idea is that making it hard  to distinguish between adjacent datasets, means making it hard  to determine whether any particular record has been exposed.
Notice that $(\epsilon, \delta)$-DP takes a ``worst-case'' approach, in that the level of ``indistinguishability'' between $(D_0, D_1)$ as defined at \Eqn{dp-1053} must hold whatever the output $y$, however unlikely its occurrence.   In common scenarios where a mechanism is repeatedly applied to the same dataset, this worst-case measurement quickly becomes severe. For example, the standard $(\epsilon,\delta)$-DP composition theorem says that ${\cal M}\circ {\cal M}$ satisfies $(2\epsilon, 2\delta)$-privacy, if ${\cal M}$ satisfies $(\epsilon,\delta)$ privacy. It is now recognised that more nuanced definitions of privacy can yield more realistic predictions of  privacy risks, with $f$-DP being a recent proposal \cite{annurev:/content/journals/10.1146/annurev-statistics-112723-034158}.

\subsection{ $f$-differential privacy}\label{f-dp-basics}

$f$-DP uses hypothesis testing as the means to distinguish between inputs.
Let  ${\cal M}(D_0)= p ~~{\cal M}(D_1)=q$; further let $H_0, H_1$ be hypotheses:
%
\[H_0: \text{the input dataset is $D_0$} \qquad H_1: \text{the input dataset is $D_1$}~.
\]
%


\begin{definition}\label{test-1047}
A test is a mapping $\phi: \mathcal{Y} \to [0,1]$ where $\phi(y) = 0$ means $H_0$ is  accepted, $\phi(y) = 1$ means $H_1$ is accepted, and $\phi(y) = c$ means $H_0$ is accepted with probability $c$. 
%
For $p$ and $q$  probability distributions, as above, the \emph{significance level of test} $\phi$ is $\alpha_\phi := \mathbb{E}_{p}[\phi]$. The \emph{power of test} $\phi$ is $1-\beta_\phi := \mathbb{E}_{q}[\phi]$.
\end{definition}

The quantities $\alpha_\phi, \beta_\phi$ are known respectively as Type I and Type II errors, or false negative and false positive rates. It turns out that the most effective test for distinguishing between distributions, in terms of minimising the false positive and negative rates, is the simple \emph{likelihood ratio test}.
The celebrated Neyman-Pearson lemma sets out the details.

\begin{lemma}[Neyman-Pearson \cite{NPL1933}]\label{NM-1511}
Let $p, q$ be distributions as in \Def{test-1047}. 
A test $\phi: \mathcal{Y} \to [0,1]$ is the most powerful at significance level $\alpha$, i.e. $\mathbb{E}_{p}[\phi] = \alpha$, if there are two constants $h \in [0,\infty]$ and $c \in [0,1]$ such that the test has the form:
\begin{equation}\label{np-form-1056}
\phi(y) = \begin{cases}
      1, & \text{if } q_y > h p_y \\
      c & \text{if } q_y = h p_y \\
      0, & \text{if } q_y < h p_y
    \end{cases}~.
    \end{equation}
\end{lemma}

A trade-off function details the relation between Type I/II errors wrt.\  most powerful tests.
\begin{definition}[Trade-off function]\label{tof-1402}
Let $p$, $q$ be as above.  The trade-off function ${\bf{T}}(p,q): [0,1] \to [0,1]$ is defined:
\begin{equation}\label{toff-1527}
{\bf T}(p,q)(\alpha) = \inf_{\phi}\{\beta_\phi: \alpha_\phi \leq \alpha\} ~.
\end{equation}
We say that $(p, q)\leq_\alpha (p', q')$ if ${\bf T}(p, q)(\alpha)\leq {\bf T}(p', q')(\alpha)$ for $0{\leq}\alpha{\leq} 1$.
\end{definition}

  It turns out that, for fixed distributions, ${\bf T}(p,q)$ is convex and satisfies ${\bf T}(p,q)(\alpha)\leq 1{-}\alpha$. The definition of $f$-DP uses \emph{abstract trade-off functions} to describe indistinguishability via hypothesis testing.

\begin{definition}[Abstract trade-off functions]\label{trade-off-1545}
The set of abstract trade-off functions $({\mathbb{F}, \leq})$ defines $f{\in}\mathbb{F}$ if it is a convex function $[0,1]{\rightarrow}[0,1]$, and satisfies $f(\alpha){\leq}(1{-}\alpha)$.  
Abstract trade-off functions are ordered pointwise, i.e\ $f{\leq} f'$ if and only if $f(\alpha){\leq} f'(\alpha)$ for all $0{\leq}\alpha{\leq}1$.

The pointwise maximum  of two abstract trade-off functions is $f\sqcup f'$ (and is a trade-off function). We define the trade-off minimum $f \sqcap f'$ to be $\sqcup\{ g~ |~ g {\leq} f~\textit{and}~ {g\leq f'}\}$. The minimum trade-off function is the constant zero function, and the maximum trade-off function  takes $\alpha$ to $1{-}\alpha$. 
\end{definition}


\begin{definition}[$f$-DP]\label{fpriv-1629}
For $f\in {\mathbb F}$, we say that mechanism $\cal{M}$ 
 satisfies $f$-DP if $f \leq {\bf T}({\cal M}(D), {\cal M}(D'))$, for any pair of adjacent datasets $D, D'$.
\end{definition}


It was shown in \cite{10.1111/rssb.12454}, that a mechanism ${\cal M}$ satisfies $(\epsilon, \delta)$-DP if and only if it  satisfies 
$f_{\epsilon, \delta}$-DP,  where $f_{\epsilon,\delta}$ is depicted in \Fig{Fig-tof-1140}, and  defined:
\begin{equation}\label{fepsdelta}
f_{\epsilon,\delta}(\alpha)= \max\{0,~~ -e^\epsilon\alpha{+} 1{-}\delta, ~~-e^{-\epsilon}\alpha + e^{-\epsilon}(1{-}\delta) \}~.
\end{equation}

\vspace{-12pt}
\begin{figure}[!th]
   \begin{minipage}{0.4\textwidth}
\includegraphics[width=0.9\textwidth]
{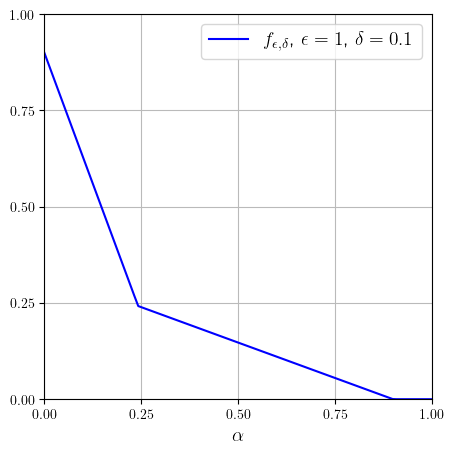}
\end{minipage}
   \begin{minipage}{0.6\textwidth}
Notice that the parameters for standard differential privacy can be read-off from the plot via gradients.
If the gradient $f_{\epsilon,\delta}$ at significance level $\alpha$ is $\epsilon'$, then it means that there is a test that can distinguish between $M_{D_0}$ and $M_{D_1}$ consistent with (pure) $\epsilon'$-DP. If 
there is some $\alpha$ for which the gradient is $0$ or $\infty$ then it means that there is some non-zero probability for which there is a test which exactly distinguishes between $M_{D_0}$ and $M_{D_1}$. For the plot at left this occurs for $1{-}\delta{\leq}\alpha{\leq}1$, where the gradient of $f_{\epsilon,\delta}$ is $0$.  
    \end{minipage}    \caption{Trade-off function $f_{\epsilon, \delta}$.  }\label{Fig-tof-1140}
\end{figure}

\vspace{-12pt}
Unlike standard differential privacy, the details of the trade-off function $f$ make it challenging to determine composition theorems that reflect $f$-DP accurately, not only for the basic composition ${\cal M}\circ {\cal M}$ mentioned above, but for other compositions that typically arise in privacy-enhancing algorithms. 
%
We set out a systematic approach for defining a range of composition theorems for $f$-DP, by demonstrating an equivalence between $({\mathbb F}, \leq)$ and a class of information channels, for which compositions can be straightforwardly defined. 

\subsection{Information channels and QIF}\label{QIFSec}
 QIF~\cite{Alvim:20a} is a framework for quantifying information leaks in programs. It features information channels as the basic model, together with the ``g-leakage'' semantics for assessing associated security risks, depending on the scenario. We set out components of QIF, and summarise its mathematical properties which are needed here.

An information channel $C$ maps inputs (secrets) $x \in \calx$ to observations $y \in \caly$ according to a distribution in $\Dist{\caly}$. In the discrete case, such channels are $\calx{\times}\caly$ matrices $C$ whose row-$x$, column-$y$ element $\CondYX$ is the probability that input $x$ produces observation $y$. The $x$-th row $C_{x,-}$ is thus a discrete distribution in $\Dist{\caly}$. 
For example the channel $M$ at ~\Eqn{M-chan} is displayed as a channel where ${\cal X}$ consists of adjacent datasets $\{ D_0, D_1\}$.

We can use Bayes rule to model an adversary who uses their observations from a channel to (optimally) update their knowledge about the secrets $\calx$. 
Given a prior distribution $\pi : \Dist{\calx}$ (representing an adversary's prior knowledge) and channel $C$, we can compute a joint distribution $J: \Dist(\calx{\times}\caly)$ where $J_{x,y} = \px \CondYX$. Marginalising down columns yields the $y$-marginals $Pr(y) = \sum_x \px \CondYX$ each having a posterior over $\calx$ corresponding to the posterior probabilities $P_{X|y}(x)$, computed as $\nicefrac{J_{x,y}}{Pr(y)}$ (when $Pr(y)$ is non-zero). We denote by $\delta^y$ the posterior distribution $P_{X|y}(X|y)$ corresponding to the observation $y$. The set of posterior distributions and the corresponding marginals can be used to compute the adversary's posterior knowledge after making an observation from the channel.

\begin{definition}[Refinement of channels]\label{ref-1814}
Let $C {\in} \calx \rightarrow \Dist{\caly}$ and $C' {\in} \calx \rightarrow \Dist{\calz}$ be channels; we say $C$ is refined by $C'$ or $C \sqsubseteq C'$ if there is a channel $W{\in} \caly\rightarrow \Dist{\calz}$ such that $C\cdot W= C'$. We call $W$ the witness to the refinement.
\end{definition}
We model information leakage using \emph{gain functions}.

\begin{definition}[Leakage semantics]\label{leak-1819}
 A gain function is a mapping ${\cal A}{\times}\calx\rightarrow {\mathbb R}$, where ${\cal A}$ is a set of actions. Given a gain function $g$, we can define a vulnerability $V_g: \Dist \calx \rightarrow {\mathbb R}$, defined $V_g[\pi]:= \max_{a\in {\cal A}}\sum_{x{\in}\calx}\pi_x{\times}g(a,x)$.
The conditional vulnerability wrt.\ channel $C$ and prior $\pi$ is given by $V_g[\pi {\triangleright} C]:= \sum_{y{\in}\caly}Pr(y)V_g[\delta^y]$.
\end{definition}
We focus on the following class of channels.
\begin{definition}
Let $(\mathbb{C}_2, \sqsubseteq)$ be the set of 2-row channels, ordered by refinement. 
\end{definition}

\noindent {\bf Summary of QIF properties established elsewhere \cite{Alvim:20a}.} Observe that the leakage semantics is based on a generalisation of the notion of ``entropy'' and we can use it to determine how much information is leaked  by comparing the prior vulnerability $V_g[\pi]$ to the posterior vulnerability $V_g[\pi {\triangleright} C]$, with the greater the difference corresponding to a greater amount of  leaked information relative to the gain $g$.  We summarise the leakage properties we need here; more details can be found elsewhere \cite{Alvim:20a}.

\begin{enumerate}[(I)]
\item $V_g[\pi {\triangleright} C]$ is independent of the column labels of $C$; this means that we can re-order the columns of $C$ without changing its leakage semantics.
\item $C \sqsubseteq C'~~ \textit{iff}~~V_g[u {\triangleright} C] \geq V_g[u {\triangleright} C']$, for all gain functions $g$, and $u$ the uniform prior on ${\cal X}$.
\item We can render a  channel $C{\in}{\mathbb C}_2$ as the corresponding ``hyper-distribution'' $[u{\triangleright} C]$ as a convex sum $\sum_{y{\in}{\cal Y}} Pr(y) \delta^y$, where we are considering $\delta^y$ as a 1-summing vector in $[0,1]{\times}[0,1]$. This means that we can depict the posteriors using a Barycentric representation, as illustrated in \Fig{Barycentric-1356}. 
\item It turns out that if $C{\in}{\mathbb C}_2$, having exactly two posteriors (i.e., over two outputs), then $C\sqsubseteq C'$ if and only if all of the posteriors of $[u{\triangleright}C']$ lie in the convex hull of the two posteriors of $[u{\triangleright}C]$.
\end{enumerate}

\begin{figure}
\centering
\begin{minipage}{0.55\textwidth}
\includegraphics[width=0.95\textwidth]{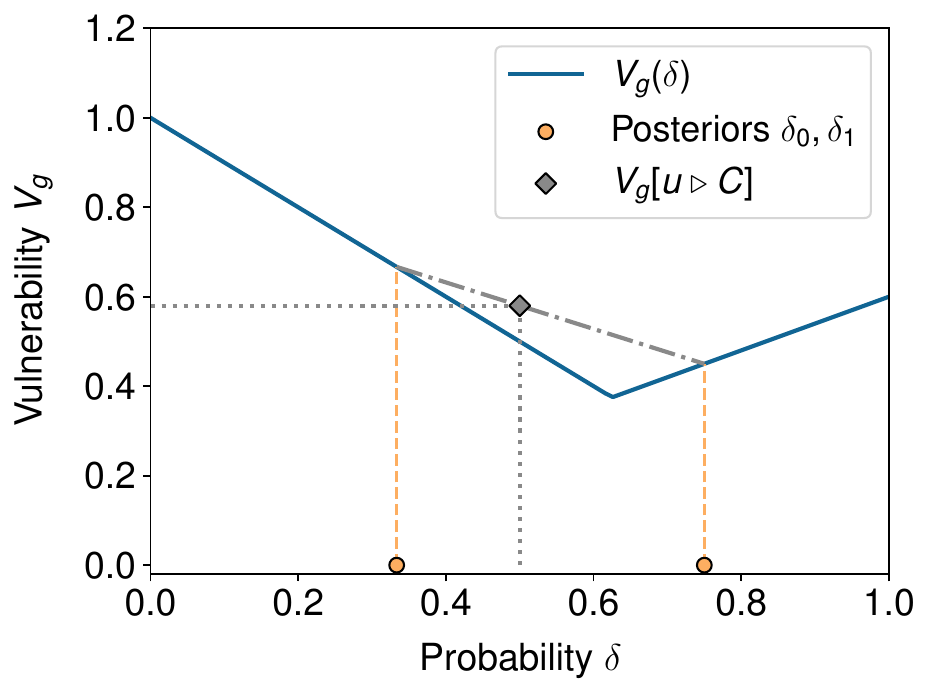}
    \end{minipage}\hfill
    \begin{minipage}{0.45\textwidth}
    \caption{Barycentric representation of $[u {\triangleright}C]$ for $C=\begin{bmatrix} 
2/5 & 3/5\\
4/5 & 1/5\end{bmatrix}$, showing posteriors $(1/3, 2/3)$ and $(3/4, 1/4)$, rendered as orange points on the horizontal axis to indicate the probability of the first component. Observe that
$V_g[u{\triangleright} C]$ corresponds to the intersection of the vertical at the mid-point, and the line connecting $V_g(\delta_0)$ and $V_g(\delta_1)$, as $C$ has only two posteriors.
}\label{Barycentric-1356}
    \end{minipage}
\end{figure}

\section{Modelling $f$-DP in QIF}\label{f-DP-QIF-Sec}
Let $M{\in}{\mathbb C}_2$ be a channel; recall that its two rows denoted by $M_{D_0}$ $M_{D_1}$ define two distributions over output ${\cal Y}$. We can therefore map two row channels to trade-off functions using \Def{tof-1402}. 

\begin{definition}[Distinguishability profile]\label{d1643}
Let $M{\in}{\mathbb C}_2$; define the distinguishability profile as a mapping ${\cal T} : \mathbb{C}_2{\rightarrow} \mathbb{F}$: 
\[
{\cal T}(M):= {\bf T}(M_{D_0}, M_{D_1})~.
\]
\end{definition}

Our aim in this section is to show that indistinguishability properties defined by trade-off functions can be modelled exactly in terms of information leakage properties of ${\mathbb C}_2$. We demonstrate that hypothesis testing at significance $\alpha$ corresponds to a class of $2{\times}2$ channels, and that the Neyman-Pearson Lemma(\ref{NM-1511}) suggests a distinguishing class of gain functions called ``hockey-stick'' functions. 

\subsubsection{\bf ${\cal T}(M)(\alpha)$ defines a refinement:} Recall the channel $M$ at \Eqn{M-chan}; by (I) above we may assume that the columns  are ordered with increasing ratio $q_i/p_i$:
%
\begin{equation}\label{orderedM}
M ~~=~~ \begin{bmatrix}[ccc|ccc]
p_0  & \dots & p_{k-1}~ & ~p_k & \dots  & p_n\\
q_0  & \dots & q_{k-1}~ & ~q_k& \dots  & q_n
\end{bmatrix}
\end{equation}
The Neyman-Pearson Lemma says that the most effective test to achieve a given significance level $\alpha$ is defined by some $h{\geq}0$; just above we have indicated a line separating the $k{-}1$'th column from the $k$'th column corresponding to such a test, where to the right of the line we have $q_r/p_r{\geq} h$, and to the left we have $q_l/p_l{<}h$. But now, referring to \Def{test-1047} the corresponding significance level $\alpha=\sum_{r\geq k}p_r$, and power ${\cal T}(M)(\alpha)=\sum_{l{<}k}q_l$. We can summarise these observations using the refinement $M^\alpha= M\cdot R^h$, where the the post-processing channel is given by:

\[R^h_{i,0} = \begin{cases}1, \quad & \text{if}\quad  q_i \geq h p_i \\
0, \quad &\text{else}.
\end{cases} \qquad \qquad R^h_{i,1} = \begin{cases}0, \quad & \text{if}\quad  q_i \geq h p_i \\
1, \quad &\text{else}~.
\end{cases}\]

More generally, we have the following direct definition.

\begin{definition}[Trade-off channel]\label{toc-1503}
Given a channel $M{\in} \mathbb{M}_2$, define the trade-off channel at significance level $\alpha$ to be:
\[
M^\alpha ~~: ~~=~~\begin{bmatrix}
1{-}\alpha ~~&~~ \alpha\\
{\cal T}(M)(\alpha) ~~&~~ 1{-}{\cal T}(M)(\alpha)
\end{bmatrix}
\]
\end{definition}

For a given $h$, we can compute the error probability for the test it defines. 

\begin{definition}[Error function]\label{err-1828}
The error probability $\alpha$ of $M$ at level $h$ is:
\[
\alpha = \textit{err}_M(h)~~=~~ \sum_{q_i-hp_i \geq 0} p_i~~\textit{and}~~ {\cal T}(M)(\textit{err}_M(h))~~=~~1{-}\sum_{q_i-hp_i \geq 0} q_i~.
\]
\end{definition}
\subsubsection{Test at level $h$ defines a gain function:} Next, for $h{\geq}0$ we can define a class of ``hockey-stick'' gain functions, so-called because they give rise to vulnerabilities that resemble a hockey stick, as illustrated in \Fig{HS-ref-805}, below.

\begin{definition}[Hockey-stick gain]\label{hsf-1254}
Given $h{\geq} 0$, we define the hockey-stick gain function $\underline{h}$:
\[
\underline{h}(a_1, d)~~:=~~ 1~~ \textit{if}~~ d=D_1~~\textit{else}~~-h~,~~~~\underline{h}(a_0, d)~~:=~~ 0~~ \textit{if}~~ d{\in}\{D_0, D_1 \}~.
\]
The associated vulnerability $V_{\underline{h}}$ is called a ``hockey-stick'' vulnerability.
\end{definition}

Finally, we can use hockey-stick functions to define a partial order on channels which, we will see below in \Thm{QIF+HT-1213}, enables us to prove an equivalence between trade-off functions and channels.

\begin{definition}[Hockey-stick order]\label{h-less-1810}
We define the hockey-stick order on channels: we say  $C \leq_h M$ whenever $V_{\underline{h}}[u \triangleright C] \leq V_{\underline{h}}[u \triangleright M]$, for all $h{\geq}0$.
\end{definition}
\vspace{-2mm}
\subsection{Trade-off functions, hockey sticks and refinement}
We illustrate, briefly, the concepts we have so far defined.
%
%
Consider channels $C, M {\in} \mathbb{C}_2$, and  recall that each  defines a trade-off function $f_C= {\cal T}(C), f_M={\cal T}(M) \in \mathbb{F}$. For a given $\alpha$, we construct the trade-off channels $C^\alpha, M^\alpha$, as per \Def{toc-1503}. An example is shown below for $\alpha = 0.1$. 
{
\renewcommand{\arraystretch}{1.5}
\setlength{\arraycolsep}{1.2ex}
\[
    M^\alpha = \begin{bmatrix}
        \nicefrac{9}{10} & \nicefrac{1}{10} \\
        \nicefrac{4}{5} & \nicefrac{1}{5}
    \end{bmatrix} \qquad
    C^\alpha = \begin{bmatrix}
        \nicefrac{9}{10} & \nicefrac{1}{10} \\
        \nicefrac{1}{2} & \nicefrac{1}{2}
    \end{bmatrix}
\]
}

In our example, $f_C(\alpha) < f_M(\alpha)$. This implies that $C^\alpha \sqsubset M^\alpha$, as depicted in \Fig{HS-ref-805} (below).
%
%
As noted in \Fig{Barycentric-1356}, computing $V_g[u{\triangleright}M]$, where $M$ has only two columns corresponds to a simple construction on the Barycentric representation. Applied here to hockey-stick functions and $M^\alpha, C^\alpha$, we can see clearly that for $C^\alpha \sqsubset M^\alpha$, the construction shows that
%
$V_{\underline{h}}
[u{\triangleright}C^\alpha] \geq V_{\underline{h}}[u{\triangleright}M^\alpha]$, indicated in \Fig{HS-ref-805} by the grey point on the orange diagonal line (corresponding to $V_{\underline{h}}[u{\triangleright}C^\alpha]$) lying above the grey point on the blue diagonal line ($V_{\underline{h}}[u{\triangleright}M^\alpha]$).

Whilst \Fig{HS-ref-805} illustrates the idea that hockey-stick gain functions characterise refinement of $2{\times}2$ channels, \Fig{Refinement-whs} shows how these observations can be transferred to refinement more generally in ${\mathbb C}_2$. The plots show two equivalent methods for computing $V_{\underline{h}}[u{\triangleright}C]$: on the left each posterior is evaluated and then averaged by their marginal, on the right the averaging happens first; by linearity the final values are the same. The trick here is to note that provided that the $h$ corresponds to the most powerful Neyman-Pearson test for the given $\alpha$, the averaging on the right corresponds to the refinement to $C^\alpha$. Therefore we can deduce that
$ V_{\underline{h}}[u \triangleright C] = \frac{1}{2} \left( \sum_{i=k}^{n} q_i - h p_i\right)$ which is also equal to $\frac{1}{2}\left(1{-}f_C(\alpha){-}h \alpha\right)$, where $\alpha =\textit{err}_C(h)$, which in turn is also equal to $V_{\underline{h}}[u \triangleright C^\alpha]$. 



\begin{figure}[!th]\label{fig:hockey_sticks-1}
\centering
    \begin{minipage}{0.6\textwidth}
\includegraphics[width=0.9\textwidth]{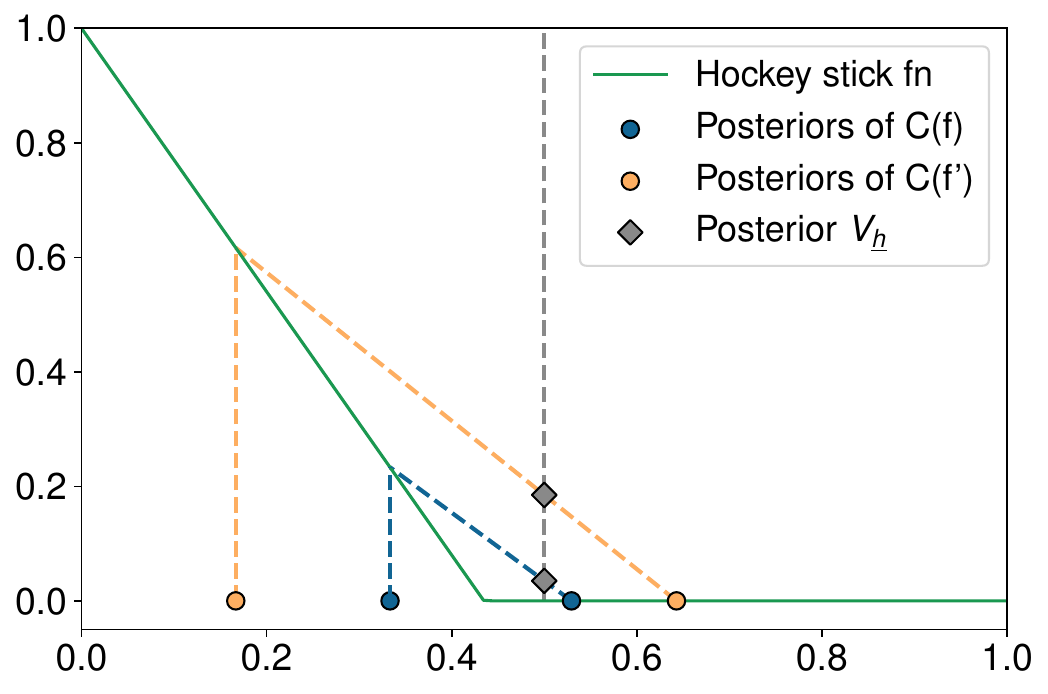}
    \label{fig:hockey_graphs}
    \end{minipage}\hfill
    \begin{minipage}{0.4\textwidth}
    \caption{Illustration of refinement: The posteriors (orange points) of $C^\alpha$ lie outside the posteriors (blue points) of $M^\alpha$ indicating by (IV) that $C^\alpha \sqsubseteq M^\alpha$. For every hockey stick function (green line), the orange diagonal line will lie above (or on) the blue diagonal line, indicating that $V_{\underline{h}}[u{\triangleright}C^\alpha] \geq V_{\underline{h}}[u{\triangleright}M^\alpha]$ for any $h$. The grey diamonds correspond to the particular $V_{\underline{h}}$ values for $C^\alpha$ and $M^\alpha$ in this example.}\label{HS-ref-805}
    \end{minipage}
\end{figure}

\begin{figure}[!th]
\includegraphics[width=\textwidth]{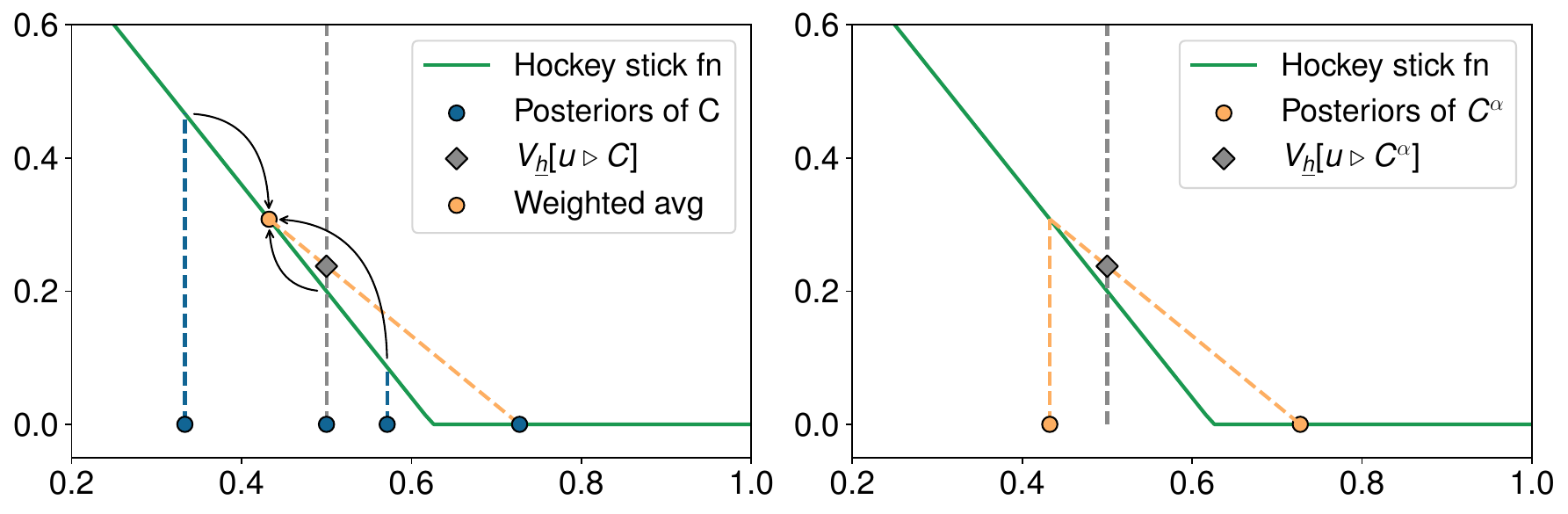}
\caption{Two equivalent methods for computing $V_{\underline{h}}[u{\triangleright}C]$: the left plot first computes $V_{\underline{h}}[\delta]$ for each (blue) posterior of $C$, and then averages, the right plot first averages the (blue) posteriors (equivalent to taking a refinement) and then computes $V_{\underline{h}}[\delta']$ on the two (orange) results.}\label{Refinement-whs}
\end{figure}

The next theorem (proved in the appendix) shows that these ideas illustrated in \Fig{HS-ref-805} and \Fig{Refinement-whs} hold in general, in particular that  hockey stick vulnerabilities are sufficient to characterise refinement in ${\mathbb C}_2$.



\begin{theorem}[Hypothesis testing in QIF]\label{QIF+HT-1213}
For $M, C{\in} \mathbb{C}_2$, the following holds:

\begin{enumerate}
\item If ${\cal T}(C)\leq {\cal T}(M)$, then $C^\alpha \sqsubseteq M^{\alpha}$~.
\item If $M^\alpha \sqsubset C^\alpha$, then there exists an $h$ such that we have $V_{\underline{h}}[u \triangleright C] < V_{\underline{h}}[u \triangleright M]$~, 
\item $C \leq_h M$ iff $M \sqsubseteq C$.
\end{enumerate}

\end{theorem}
As a corollary, we have that ${\cal T}$ is monotone.
\begin{corollary}\label{corollay1123}
If $C \sqsubseteq M$, then ${\cal T}(C) \leq {\cal T}(M)$.
\end{corollary}
\begin{proof}
Suppose, for contradiction that there is some $\alpha$ with ${\cal T}(C)(\alpha) > {\cal T}(M)(\alpha)$; this means that $M^\alpha \sqsubset C^\alpha$. Thus, by \Thm{QIF+HT-1213}(2), we can find $h$ such that $V_{\underline{h}}[u \triangleright C] < V_{\underline{h}}[u \triangleright M]$, contradicting the refinement assumption.
\end{proof}

\section{Equivalence of two-row channels and trade-off functions}\label{Galois-sec}

We show that $({\mathbb C}_2, \sqsubseteq)$ and $({\mathbb F}, \leq)$ are equivalent. We begin by defining a converse to \Def{d1643}.

\begin{definition}[Least $f$-private channel]\label{glb-1642}
Define a mapping ${\cal C}: \mathbb{F}\rightarrow {\mathbb C}_2$ as:
\[
{\cal C}(f) ~:=~ \min_{\alpha{\in}[0,1]} \begin{bmatrix}
1{-}\alpha & \alpha\\
f(\alpha) & 1{-}f(\alpha)
\end{bmatrix} ~,
\]
where $\min$ is the greatest lower bound operator in $({\mathbb C}_2, \sqsubseteq)$.
\end{definition}

\begin{lemma}[${\cal C}$ is well-defined]\label{wd-1650} Given any trade-off function $f\in {\mathbb F}$,  ${\cal C}(f)$ in \Def{glb-1642} is a well-defined channel in ${\mathbb C}_2$. (Proof in the appendix.)
\end{lemma}


As an example, if $\alpha_0 < \alpha_1$ then the greatest lower bound is given by:
\[
\begin{bmatrix}
1{-} \alpha_0 & \alpha_0\\ 
f(\alpha_0) & 1{-} f(\alpha_0)
\end{bmatrix}\min
\begin{bmatrix}
1{-} \alpha_1 & \alpha_1\\ 
f(\alpha_1) & 1{-} f(\alpha_1)
\end{bmatrix} ~~ =~~ 
\begin{bmatrix}
1{-} \alpha_1 & \alpha_1{-}\alpha_0   &\alpha_0\\ 
f(\alpha_1) &  f(\alpha_0){-}f(\alpha_1)&1{-} f(\alpha_0)
\end{bmatrix}~.
\]
In the appendix we show how this definition extends more generally in ${\mathbb C}_2$. Meanwhile, we have the following Galois connection.
\begin{theorem}[Galois connection]\label{GC-1651}
Let $f{\in}\mathbb{F}$ and $M{\in} \mathbb{C}_2$, then:
\begin{equation}
f \leq {\cal T}(M)  ~~ \Leftrightarrow ~~ {\cal C}(f) \sqsubseteq M~.
\end{equation}
\begin{proof}(Sketch)
Assume that $f\leq {\cal T}(M)$, then by \Thm{QIF+HT-1213} part (1), 
\begin{equation}\label{refinemnet-alpha1820}
f^\alpha ~~=~~\begin{bmatrix}
1{-}\alpha & \alpha\\
f(\alpha) & 1{-}f(\alpha)
\end{bmatrix} ~~~ \sqsubseteq~~~
\begin{bmatrix}
1{-}\alpha & \alpha\\
{\cal T}(M)(\alpha) & 1{-}{\cal T}(M)(\alpha)
\end{bmatrix}~~=M^\alpha~~, 
\end{equation}
hence ${\cal C}(f) \sqsubseteq {\cal C}({\cal T}(M))= M$~. See Appendix \Sec{App-B} for the last equality.

Now, suppose that ${\cal C}(f)  \sqsubseteq M$ (so that for all hockey-stick functions we must have $V_{\underline {h}}[u\triangleright M]{\leq}V_{\underline {h}}[u \triangleright {\cal C}(f)]$). Assume by contradiction that there is some $\alpha$ such that $f(\alpha) {>} {\cal T}(M)(\alpha)$. 
This means that 
$M^\alpha \sqsubset f^\alpha$, with the refinement being strict. Then using \Thm{QIF+HT-1213} part (2),  there must be an $h$ such that
\[
V_{\underline {h}}[u \triangleright {\cal C}(f)] < V_{\underline {h}}[u \triangleright M]~.
\]
which is a contradiction of the assumption ${\cal C}(f)  \sqsubseteq M$. 
\end{proof}
\end{theorem}
As a corollary, we have that the greatest lower bound (glb) operator in ${\mathbb C}_2$ corresponds to the lattice minimum of ${\mathbb F}$.
\begin{corollary}\label{Samemins-1040}
Let $C, C' \in {\mathbb C}_2$. Then ${\cal T}(C \min C') = {\cal T}(C) \sqcap {\cal T}(C')$.
\begin{proof}
 Since $C\min C' \sqsubseteq C, C'$, it follows by \Cor{corollay1123} that ${\cal T}(C \min C')\leq {\cal T}(C)$ and ${\cal T}(C \min C')\leq {\cal T}(C')$. Hence, ${\cal T}(C \min C')\leq {\cal T}(C)\sqcap {\cal T}(C')$.  

 Next, since ${\cal T}(C)\sqcap {\cal T}(C') \leq {
 \cal T}(C), {
 \cal T}(C')$, by \Thm{GC-1651}, we must have ${\cal C}({\cal T}(C)\sqcap {\cal T}(C'))\sqsubseteq C, C'$, and therefore by the glb property,  ${\cal C}({\cal T}(C)\sqcap {\cal T}(C'))\sqsubseteq C\min C'$ also. Finally, appealing to \Thm{GC-1651} again, it follows that ${\cal T}(C)\sqcap {\cal T}(C') \leq {\cal T}(C \min C')$.
\end{proof}
\end{corollary}

\subsection{Finite channels and piecewise linear trade-off functions.}
It turns out that when $M$ has a finite number of columns, ${\cal T}(M)$ is \emph{piecewise linear}, which means that the domain $[0,1]$ can be split into finitely many disjoint sub-intervals, such that ${\cal T}(M)$ is linear on each sub-interval.
We denote the subset of piecewise linear trade-off functions by ${\mathbb F}^{PL}$; for each such $f$, there are correspondingly a finite set of increasing $\alpha_i$, ($0{\leq} i{\leq} n$) such that $f$ is linear on the sub-intervals $[\alpha_i, \alpha_{i{+}1}]$. We call these $\alpha_i$ ``facet'' points of $f$. Next, we can show directly that when $\alpha_i {\leq}\alpha{\leq}\alpha_{i+1}$ that $(f^{\alpha_i}\min f^{\alpha^{i+1}})\sqsubseteq f^{\alpha}$ (where $f^\alpha$ is as at \Eqn{refinemnet-alpha1820}), and therefore we deduce that  ${\cal C}(f)$ is determined by the facet points for $f$. This then supports \Alg{channel-alg-1734} for computing ${\cal C}(f)$ by forming the greatest lower bound of increasing $\alpha_i$, i.e.:
$
{\cal C}(f) = (f^{\alpha_0} \min f^{\alpha_1})\min f^{\alpha_2})\cdots \min f^{\alpha_n}))~.
$


%
%

\begin{algorithm}
\begin{algorithmic}\caption{Computing the channel $\mathcal{C}(f)$ for a given $f{\in}{\mathbb F}^{PL}$}\label{channel-alg-1734}
\Require $f\in {\mathbb F}$; $0=\alpha_0{<}\dots {<}\alpha_{n}=1$ correspond to the facet points of $f$
\Ensure Channel $\begin{bmatrix}C_{-0} & \dots C_{-n} \end{bmatrix}\in {\mathbb C}_2$ such that $C$ is equal to $\mathcal{C}(f)$
\State $C_{-,n} \gets \begin{bmatrix} 
\alpha_{0}\\
1{-}f(\alpha_{0})
\end{bmatrix}$ \Comment{Compute the last column of of ${\cal C}(f)$}
\State $i\gets 0$;
\While{$i<n$} 
\State $C_{-,n{-}i{-}1} \gets \begin{bmatrix} 
\alpha_{i+1}{-}\alpha_{i}\\
f(\alpha_i){-}f(\alpha_{i+1})
\end{bmatrix}$ \Comment{Compute the $n{-}i{-}1$'th column}
\State $i\gets i+1;$
\EndWhile
\end{algorithmic}
\end{algorithm}

To see \Alg{channel-alg-1734} in action,  recall the trade-off function $f_{\epsilon,\delta}$ depicted in \Fig{Fig-tof-1140}, which we observe is in ${\mathbb F}^{PL}$ with facet points $\alpha_0{=}0, ~\alpha_1{=}\frac{1-\delta}{e^\epsilon+1}, ~\alpha_2{=}1{-}\delta, ~\alpha_3=1$.  \Alg{channel-alg-1734} computes the columns of ${\cal C}(f_{\epsilon, \delta})$ in order of increasing $\alpha_i$ as follows:
\[
\begin{bmatrix}
1{-}(1{-}\delta)~~&
(1{-}\delta){-} (\frac{1-\delta}{e^\epsilon+1})~~ &~~ 
 \frac{1-\delta}{e^\epsilon+1}{-}0 ~~&0\\
 f_{\epsilon,\delta}(1{-}\delta){-}f_{\epsilon,\delta}(1)    ~~&~~ 
f_{\epsilon,\delta}(\frac{1-\delta}{e^\epsilon+1}){-}f_{\epsilon,\delta}(1{-}\delta)~~&~~
f_{\epsilon,\delta}(0){-}f_{\epsilon,\delta}(\frac{1-\delta}{e^\epsilon+1})  ~~&1{-}f_{\epsilon,\delta}(0)
\end{bmatrix}~,
\]
yielding,

\begin{equation}\label{Canonical-epsdelta}
{\cal C}(f_{\epsilon,\delta})= C_{\epsilon, \delta} ~~=~~ 
\begin{bmatrix}
 \delta~~&{(1{-}\delta)e^\epsilon}/{1{+}e^\epsilon} ~~ &~~ {(1{-}\delta)}/{1{+}e^\epsilon} ~~&0\\
0~~&~~ {(1{-}\delta)}/{1{+}e^\epsilon} ~~&~~ {(1{-}\delta)e^\epsilon}/{1{+}e^\epsilon} ~~&\delta
\end{bmatrix}~.
\end{equation}


By \Thm{GC-1651} this is the greatest lower bound in the $\sqsubseteq$ order for channels in ${\mathbb C}_2$ that satisfy $f_{\epsilon, \delta}$-DP. Thus we have the following corollary:
\begin{corollary}[Canonical $(\epsilon,\delta)$ channel]\label{cor:eps:delta:refined:by:M}
Let $M{\in}{\mathbb C}_2$. Then $M$ satisfies $(\epsilon, \delta)$-DP if and only if $C_{\epsilon,\delta} \sqsubseteq M$.
\end{corollary}



\section{Compositions of channels define compositions for $f$-DP}\label{comp-defs}

With \Thm{GC-1651}, we can now obtain composition rules for $f$-DP, by using compositions defined on channels \cite{Alvim:20a}. We  do this for typical compositions used to implement or analyse privacy mechanisms.

\subsection{Parallel composition}
A typical scenario for analysis is repeated application of a mechanism ${\cal M}$ to the same dataset.  This assumes that the output of ${\cal M}\circ {\cal M}(D)$ is a pair$(y_0, y_1)$, one for each (independent) application of ${\cal M}$. 
In ${\mathbb C}_2$, this corresponds to parallel composition: if  $C: {\cal X}\rightarrow{\cal Y}$  and $M: {\cal X}\rightarrow{\cal Z}$, 
then the parallel composition $C\parallel M$ outputs a pair from ${\cal Y}{\times}{\cal Z}$ as follows \cite{Alvim:20a}:
\begin{equation}\label{parallel-eqn}
(C\parallel M)_{x,(y,z)}=  C_{x,y}{\times} M_{x, z} ~.
\end{equation}

%
%
We can obtain an exact privacy profile for parallel by using \Thm{GC-1651} to express channels as the $\min$ of their trade-off channels (\Def{toc-1503}), and then \Cor{Samemins-1040} which says that $\min$ in ${\mathbb C}_2$ corresponds to $\sqcap$ in ${\mathbb F}$.  Let ${\cal T}(C)\geq f$, and ${\cal T}(M) \geq f'$,   then: 
%

\begin{equation}\label{parallel-1842}
{\cal T}(C\parallel M) ~\geq~ \bigsqcap_{\alpha, \alpha'} {\cal T}\left(\begin{bmatrix} 1{-}\alpha & \alpha\\
f(\alpha) & 1{-}f(\alpha)\end{bmatrix} ||  
 \begin{bmatrix} 1{-}\alpha' & \alpha'\\
f'(\alpha') & 1{-}f'(\alpha')\end{bmatrix} \right)~.
\end{equation}


As an example, 
\Fig{Composition-1426}, illustrates how the $f$-DP rule gives a better measurement for privacy loss than does the standard $(\epsilon,\delta)$-DP composition rule.

\begin{figure}
\centering
\begin{minipage}{0.6\textwidth}
\includegraphics[width=0.9\textwidth]{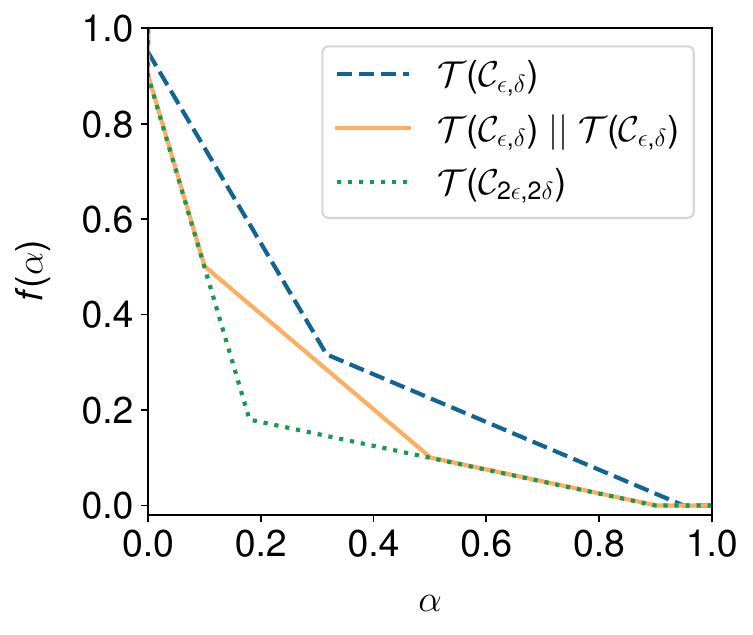}
\end{minipage}\hfill
\begin{minipage}{0.4\textwidth}
\caption{Detailed privacy profiles for composition. Notice how orange plot lies above the green plot, and has a gradient of $-1$ for $\alpha$ between $0.1$ and $0.4$ (approx.) indicating no privacy leakage for these tests, unlike the green plot which corresponds to the standard DP-composition, giving an over approximation for privacy loss. }\label{Composition-1426}
\end{minipage}
\end{figure}

\subsection{Visible probabilistic choice}\label{vis choice}
Visible probabilistic choice models the situation where the
choice between applying $M$ or $C$ is made randomly, but which one was applied can be deduced from the output. This occurs when the outputs of the two channels are drawn from non-intersecting sets.
Given channels $C:{\cal X}\rightarrow {\cal Y},M:{\cal X}\rightarrow {\cal Z}$, where ${\cal Y}\cap {\cal Z}= \phi$ we define the \emph{visible} probabilistic choice $C \leftindex_r{\oplus} M$ as follows:
\[
\begin{bmatrix}
p_0 & \dots  & p_n\\
q_0  & \dots  & q_n
\end{bmatrix}
\leftindex_r{\oplus}
\begin{bmatrix}
p'_0 & \dots  & p'_m\\
q'_0  & \dots  & q'_m
\end{bmatrix} ~~~ =~~~
\begin{bmatrix}
rp_0 & \dots  & rp_n & (1{-}r)p'_0 & \dots  & (1{-}r)p'_m\\
rq_0  & \dots  & rq_n & (1{-}r)q'_0  & \dots  & (1{-}r)q'_m
\end{bmatrix} ~.
\] 

In the (abstract) channel at right, all columns are derived from either $C$ or $M$, but scaled by $r$ or $1{-}r$ depending on whether the column originated from $C$ or $M$. It turns out that the visible probabilistic choice composition rule for $f$-DP is determined by the error function \Def{err-1828}. As before, let ${\cal T}(C)\geq f$, and ${\cal T}(M) \geq f'$,  and let
 $\alpha= \textit{err}_{{\cal C}(f)}(h)$, and  $\alpha' = \textit{err}_{{\cal C}(f')}(h)$, then:
\begin{equation}\label{vis-choice}
{\cal T}(C \leftindex_r\oplus M)(r{\times}\alpha+(1{-}r){\times}\alpha') ~\geq~ r{\times}f(\alpha) + (1{-}r){\times}f'(\alpha')~.
\end{equation}


As an example, we notice that the canonical channel $C_{\epsilon, \delta}$ is actually the visible probabilistic choice of two channels, i.e. $C_{\epsilon, \delta}= C_{\infty, 0} \leftindex_\delta\oplus~  C_{\epsilon, 0}$, where $C_{\infty,0}$ is the extreme channel that reveals exactly what the inputs are, and $C_{\epsilon, 0}$ is known as a ``pure random response channel''. Using (\ref{vis-choice}) we note that
for $h{\geq} 0$, we have $\textit{err}_{C_{\infty,0}}(h)=0$ and ${\cal T}(C_{C_{\infty,0}})(0)=0$, 
thus  we have immediately that ${\cal T}(C_{\epsilon, \delta})(\delta{\times}0+(1{-}\delta){\times}\alpha)={\cal T}(C_{\epsilon, \delta})((1{-}\delta)\alpha)=(1{-}\delta){\cal T}(C_{\epsilon,0})(\alpha)$.
    \Fig{Vis choice}, in the appendix, illustrates more examples. 

\subsection{Hidden probabilistic choice}


Hidden probabilistic choice is similar to visible probabilistic choice in that the choice between applying $M$ or $C$ is made randomly, but \emph{unlike} visible choice it \emph{cannot} be determined by looking at the outputs which one was applied. This situation occurs in implementations such as \Alg{Purify-1511}, as described in \Sec{Applications-main}.



For example, assume that $C,M$ both have $n{+}1$ columns with the columns labelled with outputs ${\cal Y}= \{y_0,\dots y_n\}$. The hidden probabilistic choice with bias $r$ produced by combining $C$ and $M$ in this way is:
\[
\begin{bmatrix}
p_0 & \dots  & p_n\\
q_0  & \dots  & q_n
\end{bmatrix}
\leftindex_r{\boxplus}
\begin{bmatrix}
p'_0 & \dots  & p'_n\\
q'_0  & \dots  & q'_n
\end{bmatrix} ~~~ =~~~
\begin{bmatrix}
rp_0 {+} (1{-}r)p'_0& \dots  & rp_n{+} (1{-}r)p'_n\\
rq_0{+} (1{-}r)q'_0  & \dots  & rq_n{+} (1{-}r)q'_n 
\end{bmatrix} ~.
\]
Interestingly, the $f$-DP composition rule for hidden probabilistic choice does not have a direct definition because hidden probabilistic choice is sensitive to the precise outputs of the two mechanisms, and this information is not recorded in trade-off functions. However it is still the case that we can compute the indistinguishability profile for algorithms that use hidden probabilistic choice, by applying ${\cal T}$ to the channel composition, as in right-hand-side channel above.

\subsection{Pre-processing as a composition}
Our final composition is pre-processing, which arises when the inputs to a privacy mechanism are processed in some way before being presented to the mechanism. This is the case for the common example of sub-sampling in machine learning applications \cite{Balle_Barthe_Gaboardi_2020}.
Given a mechanism $M$, a pre-processing process $P$  is applied before applying the mechanism $M$.  As a channel, this is modelled as a pre-matrix multiplication $P\cdot M$, thus its privacy profile becomes ${\cal T}(P \cdot M)$.



\section{Universal Properties of Compositions}\label{Univ-comps-1028}

Finally we study relationships between privacy profiles of the different operators.

\begin{theorem}[Composition theorems]\label{composition-results-1310}
The following inequalities hold:
\begin{enumerate}
\item (Parallel composition) ${\cal T}(C\parallel D)\leq {\cal T}(C) \sqcap {\cal T}(D)$
\item (Visible choice) ${\cal T}(C) \sqcap {\cal T}(D) \leq {\cal T}(C \leftindex_p\oplus D) \leq {\cal T}(C) \sqcup {\cal T}(D)$
\item (Pre-processing) ${\cal T}(C) \leq {\cal T}(C \cdot Q)$
\item (Visible and hidden choice) ${\cal T}(C \leftindex_p \oplus D) \leq {\cal T}(C \leftindex_p \boxplus D)$
\item Let $\mathcal Y_1, \mathcal Y_2$ be two disjoint sets. Let $D,E: \mathcal X \to \mathcal Y_1$ and $C: \mathcal X \to \mathcal Y_2$. Then 
\[{\cal T}((C \leftindex_p \oplus D)\leftindex_r \boxplus E) = {\cal T}(C \leftindex_{rp} \oplus (D \leftindex_{\frac{r(1-p)}{r(1-p)+(1-r)}} \boxplus E))\]
\begin{proof}
All results follow from well-known channel refinements and \Thm{GC-1651} and \Cor{corollay1123}. For example, (1) follows since $C\parallel D\sqsubseteq C,D$. For (2) we have, 
\[
(C \min D) = (C \min D) \leftindex_p\oplus  (C \min D)\sqsubseteq C \leftindex_p\oplus D~, 
\]
with the last inequality following from the observation that $C\min D \sqsubseteq C, D$, and $\leftindex_p\oplus$ is a monotone operator. The remaining inequalities follow similarly.
\end{proof}
\end{enumerate}
\end{theorem}

We also have the following monotonicity results, which again follow from \Thm{GC-1651}, \Cor{corollay1123} and refinement properties of channels \cite{Alvim:20a}.

\begin{theorem}[Monotonicity results]\label{Refinements-1523}
Let $C\sqsubseteq C'$ and $Q\sqsubseteq Q'$, then the following refinements hold:

\begin{enumerate}
\item (Parallel Composition) ${\cal T}(C\parallel Q) \leq {\cal T}(C'\parallel Q)$
\item (Visible choice) ${\cal T}(C \leftindex_p\oplus Q) \leq {\cal T}(C' \leftindex_p\oplus Q)$
\item (Pre-processing) ${\cal T}(C \cdot Q)\leq {\cal T}(C \cdot Q')$
\end{enumerate}
\end{theorem}



We end this section by demonstrating a canonical representation for a common form of trade-off function, namely symmetric, piecewise linear functions.

\begin{definition}[Symmetric, piecewise linear trade-off functions]\label{Sym-trade-1214}
A trade-off function is symmetric if $f(\alpha)= f\circ f(\alpha)$, for all $0{\leq}\alpha{\leq}1$. It is finite, piecewise linear  if it is linear almost everywhere, except for  a finite number of facet points.  
\end{definition}

Recall $f_{\epsilon, \delta}$ which is piecewise linear and symmetric; we noted above that there is a canonical representation of ${\cal C}(f_{\epsilon, \delta})$ as a visible probabilistic choice over channels $C_{\epsilon, 0}$ and $C_{\infty, 0}$. It turns out that this is true generally for all symmetric, piecewise linear trade-off functions.

\begin{lemma}[Canonical representation for symmetric trade-off functions]\label{Can-rep-1220}
Let $f$ be a finite, symmetric trade-off function, with $N{+}1$ facet points. Then there are $\epsilon_0, \dots \epsilon_{\lfloor N/2\rfloor}$ reals such that ${\cal C}(f)$ is a visible probabilistic choice over $C_{\epsilon_i,0}$.
\begin{proof}
Direct consequence of \Alg{channel-alg-1734}, using the symmetric condition of $f$. In particular, let $M$ be the result of applying \Alg{channel-alg-1734}. This means that $M$ has  $N{+}1$ columns and two rows, such that $M_{x,n}= M_{1{-}x, N-n}$. We define $\epsilon_n= M_{0,n}/M_{1,n}$ for $0{\leq} n {\leq} N/2$, and observe that the two columns at $M_{-n}, M_{-(N{-}n)}$ correspond to $C_{\epsilon_n, 0}$, scaled by (visible) probability $M_{0,n}{+}M_{0,(N{-}n)}$.
\end{proof}
\end{lemma}

We note finally that many of our analyses can be usefully carried out using the canonical mechanisms satisfying $f$-DP, since the monotonicity results of \Thm{Refinements-1523} imply a (tight) lower bound for the class of mechanisms satisfying the given $f$-DP constraint.  
As usual, refinement ensures that a property holds generally.

\section{Implementations of privacy-enhancing mechanisms}\label{Applications-main}
\subsubsection{Privacy purification} 
We can now use the operators in \Sec{comp-defs} to describe the privacy semantics of \Alg{Purify-1511} in ${\mathbb C}_2$, and then use \Thm{GC-1651} to compute a tight privacy profile. Recall that the algorithm takes as input a mechanism $M$ that we assume satisfies some $f_{\epsilon, \delta}$-differential privacy specification, and the objective is to  ensure that the profile of the output corresponds to a pure differential privacy constraint. In this case it means that ${\cal T}(\Alg{Purify-1511})$ has gradient always bounded away from $0$ and $\infty$. This is equivalent to $f_{\epsilon', 0}\leq {\cal T}(\Alg{Purify-1511})$, for some $\epsilon'{\geq} 0$. We first show how to model \Alg{Purify-1511} using channel compositions.

We observe that the first six lines of \Alg{Purify-1511} correspond to a probabilistic choice with bias $r$ i.e.\ $M$ (is applied to $x$) with probability $r$, or, with probability $1{-}r$, $x$ is ignored and a random value is reported. This choice is a combination of hidden and visible choice depending on overlap of  ${\cal Y}$  and ${\cal Y'}$.
\footnote{We observe that for well-definedness ${\cal Y'}$ must be bounded.}
When they overlap entirely, this is a hidden choice; when they overlap partially then it is a combination of visible and hidden choice. For example, if we assume that ${\cal Y'}\subset {\cal Y}$, then we can rewrite $M$ as $M_1\leftindex_p\oplus M_2$, where the output of $M_1$ is ${\cal Y'}$; we can then use \Thm{composition-results-1310}(5). to compute the trade-off function of $M\leftindex_r\boxplus U[{\cal Y'}]$.
  


The final line of the algorithm, before the output is effectively a post-processing by a Geometric perturbation, which we denote by $G_{\epsilon'}$. 
Overall we can model the effect of the purification algorithm as $(M \leftindex_r\boxplus U[{\cal Y}'])\cdot G_{\epsilon'}$.

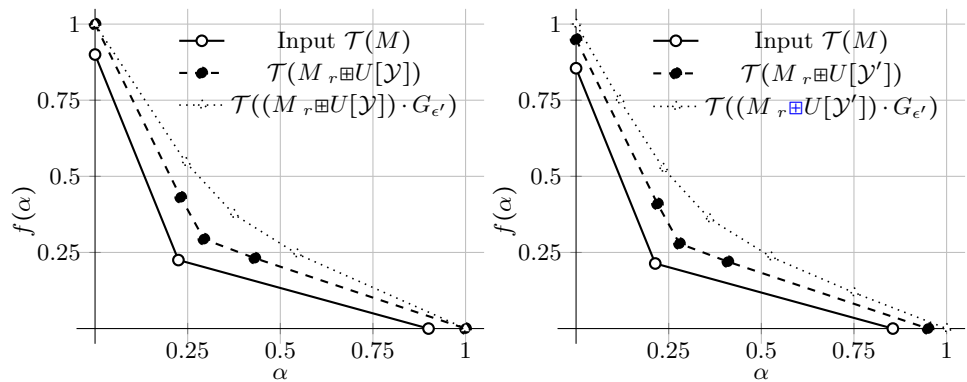
\begin{figure}[!th]\label{fig:hockey_sticks}
\begin{tikzpicture}
  \begin{axis}[
    width=7cm,
    axis lines=middle,
    axis line style={-}, 
    xmin=-0.05, xmax=1.05,
    ymin=-0.05, ymax=1.05,
    xtick={0,0.25,0.5,0.75,1},
    ytick={0,0.25,0.5,0.75,1},
    grid=both,
    xlabel={$\alpha$},
    ylabel={$f(\alpha)$},
    xlabel style={
      at={(axis description cs:0.5,0)},
      anchor=north,
      yshift=-6pt
    },
    ylabel style={
      at={(axis description cs:0,0.5)},
      anchor=east,
      xshift=-20pt,rotate=90
    },
    enlargelimits=false,
    clip=false,
    legend style={
      at={(0.97,0.97)}, 
      anchor=north east,
      draw=none,
      fill=none
    }
  ]

      \addplot[thick, mark=*, mark options={fill=white}] coordinates {
      (0,0.9)
      (0.225,0.225)
      (0.9,0)
    };
    \addlegendentry{Input ${\cal T}(M)$ }
    \addplot[ thick, dashed, mark=*] coordinates {
      (0,1)
      (0.23125,0.43125)
      (0.29375,0.29375)
      (0.43125,0.23125)
      (1,0)
    };
    \addlegendentry{${\cal T}(M\leftindex_r\boxplus U[{\cal Y}])$ }

    \addplot[semithick, dotted, mark=triangle*, mark options={fill=white}] coordinates {
      (0,1)
      (0.2453,0.5453)
      (0.3750,0.3750)
      (0.5453,0.2453)
      (1,0)
    };
    \addlegendentry{${\cal T}((M\leftindex_r\boxplus U[{\cal Y}])\cdot G_{\epsilon'})$}    
  \end{axis}
\end{tikzpicture}%
%
\begin{tikzpicture}
  \begin{axis}[
    width=7cm,
    axis lines=middle,
    axis line style={-}, 
    xmin=-0.05, xmax=1.05,
    ymin=-0.05, ymax=1.05,
    xtick={0,0.25,0.5,0.75,1},
    ytick={0,0.25,0.5,0.75,1},
    grid=both,
    xlabel={$\alpha$},
    ylabel={$f(\alpha)$},
    xlabel style={
      at={(axis description cs:0.5,0)},
      anchor=north,
      yshift=-6pt
    },
    ylabel style={
      at={(axis description cs:0,0.5)},
      anchor=east,
      xshift=-20pt,rotate=90
    },
    enlargelimits=false,
    clip=false,
    legend style={
      at={(0.97,0.97)}, 
      anchor=north east,
      draw=none,
      fill=none
    }
  ]
      \addplot[thick, mark=*, mark options={fill=white}] coordinates {
      (0,0.855)
      (0.2137,0.2137)
      (0.855,0)
    };
    \addlegendentry{Input ${\cal T}(M)$ }
    \addplot[ thick, dashed, mark=*] coordinates {
      (0,0.95)
      (0.2197,0.4097)
      (0.2791,0.2791)
      (0.4097,0.2197)
      (0.95,0)
    };
    \addlegendentry{${\cal T}(M \leftindex_r\boxplus U[{\cal Y'}])$}
  
    \addplot[semithick, dotted, mark=triangle*, mark options={fill=white}] coordinates {
      (0,1)
      (0.1176,0.7507)
      (0.2352,0.5265)
      (0.3605,0.3605)
      (0.5265,0.2352)
      (0.7507,0.1176)
      (1,0)
    };
    \addlegendentry{${\cal T}((M \leftindex_r{\textcolor{blue}\boxplus}{U[{\cal Y'}])\cdot G_{\epsilon'}})$}
 \end{axis}
\end{tikzpicture}

 \caption{Results at left assume ${\cal Y}'= {\cal Y}$, and at right that ${\cal Y'}\subset{\cal Y}$}\label{purify-left-right}.   
\end{figure}

\vspace{-12pt}
In \Fig{purify-left-right} we consider the two scenarios when ${\cal Y}$ and ${\cal Y}'$ do, or when they don't, coincide. We assume that $M=C_{\epsilon, \delta}$, and that each of the four columns of $M$ are selected uniformly by $U[{\cal Y}]$. The left-hand plot shows the profiles ${\cal T}(M$), ${\cal T}(M\leftindex_r\boxplus U[{\cal Y}])$, and ${\cal T}((M\leftindex_r\boxplus U[{\cal Y}])\cdot G_{\epsilon'})$. We see that both steps of \Alg{Purify-1511} improve the privacy of $M$; but when ${\cal Y}={\cal Y}'$, the final post-processing is not necessary for purification since ${\cal T}((M\leftindex_r\boxplus U[{\cal Y}])$ has finite, non-zero gradient.

This contrasts with the right-hand plot where we assume that ${\cal Y'}\subset {\cal Y}$. In this case the plot ${\cal T}(M\leftindex_r{\textcolor{blue}\boxplus} U[{\cal Y}'])$  crosses the horizontal axis at a point strictly less than 1 so that there is a (non-trivial) gradient of $0$. 
This means that the final post-processing step is crucial to ensure that the the overall sanitisation satisfies the desired pure differential privacy property. 

Interestingly, if we know that the input mechanism $M$ commutes with the final line of \Alg{Purify-1511}
using our composition rule, we can deduce the following:
\[
G_{\epsilon'} \sqsubseteq M \cdot G_{\epsilon'} \sqsubseteq ((M\leftindex_r\boxplus U[{\cal Y}])\cdot G_{\epsilon'}~,
\]
and therefore $f_{\epsilon',0}={\cal T}(G_{\epsilon'})\leq {\cal T}(((M\leftindex_r\boxplus U[{\cal Y}])\cdot G_{\epsilon'})$, leading to the ensures clause of \Alg{Purify-1511} being tightened to ``Satisfies $f_{\epsilon',0}$-DP''. This happens when $M$ is a simple perturbation, i.e.\ $x \gets x+ \mu(0)$, where $\mu$ is some probability distribution.

\subsubsection{Sub-sampling}
Sub-sampling is a technique to implement private training in machine learning \cite{steinke2022compositiondifferentialprivacy}, and \Alg{Subsample-poisson-2124} is an example.  A private dataset $D$ is to be used for training using a privacy-preserving process $M$. Rather than using the entire dataset $D$, only a sample $d\subseteq D$ is used as input to $M$.  The while-loop in \Alg{Subsample-poisson-2124} implements probabilistic sampling known as ``Poisson sampling'', where each record in $D$ is included in the sample $d$ with probability $\gamma$. 

\begin{algorithm}
\begin{algorithmic}\caption{Privacy-preserving with sub-sampling}\label{Subsample-poisson-2124}
\Require Mechanism $M:{\cal D}\rightarrow{\cal Y}$, Private dataset $D$; Parameters: $\gamma\in [0,1]$
\Ensure $z$ is a sanitised output for $D$, with better privacy than $M$ without sub-sampling. 
\State $i, d \gets |D|, \phi$;
\While ~ $i>0$
\State $v\gets U[0,1]$; \Comment{Choose a value uniformly from $[0,1]$}
\If {$(v< \gamma)$} \Comment{ Probabilistic choice with $\gamma$ bias, used as a pre-processor}
\State  $d\gets d \cup D_i$ 
\EndIf
\State $i\gets i{-}1$;
\EndWhile
\State $z\gets M(d)$; \Comment{Apply the mechanism to the sample $d$}
\State Output $z$; \Comment{Output sanitised result}
\end{algorithmic}
\end{algorithm}

\begin{figure}[!th]
    \begin{minipage}{0.5\textwidth}
\begin{tikzpicture}
  \begin{axis}[
    width=7cm,
    axis equal image,
    axis lines=left,
    xlabel={$\alpha$},
    ylabel={$f(\alpha)$},
    xmin=0, xmax=1,
    ymin=0, ymax=1,
    grid=both,
    legend pos=north east,
    xtick={0,0.25,0.5,0.75,1},
    ytick={0,0.25,0.5,0.75,1}
  ]

  \addplot[
    color=blue,
    thick,
    mark=*,
    mark options={fill=blue}
  ]
  coordinates {
    (0, 0.9)
    (0.24198, 0.24198)
    (0.9, 0)
  };
  \addlegendentry{ ${\cal T}(M)$
  }

  \addplot[
    color=red,
    thick,
    mark=square*,
    mark options={fill=red}
  ]
  coordinates {
    (0, 0.98)
    (0.24198, 0.654)
    (0.9, 0.08)
    (1, 0)
  };
  \addlegendentry{ ${\cal T}(P\cdot M)$
  }
  \end{axis}
\end{tikzpicture}
\end{minipage}
\begin{minipage}{0.6\textwidth}
Thus sub-sampling becomes a pre-processing of $M$; for Poisson sampling it can be shown that the pre-processing channel
$
P= \begin{bmatrix}
1 & 0\\
1{-}\gamma &\gamma 
\end{bmatrix}
$, so that the privacy profile for \Alg{Subsample-poisson-2124} is ${\cal T}(P\cdot M)$.  Moreover, by \Thm{Refinements-1523}(3), any mechanism $C_{\epsilon, \delta}\sqsubseteq M'$ satisfies ${\cal T}(P\cdot C_{\epsilon, \delta})$-differential privacy, defined by the red line in the plot. Note that other methods of sub-sampling, such as sampling with replacement, can be modelled in a similar way, although the details of the pre-processing is different.
\end{minipage}
\caption{Privacy profile for \Alg{Subsample-poisson-2124}\label{Sub-sample-757}}
\end{figure}

\Fig{Sub-sample-757} compares the effect of pre-processing with $P$. If we know that $M$ satisfies $f_{\epsilon,\delta}$, then the red plot shows the overall privacy profile.

\section{Related Work}
$f$-DP was introduced by Dong et al.~\cite{10.1111/rssb.12454} as a way to obtain more accurate analyses of privacy, showcasing the benefits largely for the Gaussian perturbation mechanism. It is also finding success for auditing privacy in complex mechanisms \cite{10.5555/3620237.3620329}.  Much of the definitions from \Sec{f-dp-basics} are defined in that work. Dong et al.\ showed that for every trade-off function $f$ there exists a distribution $\mu$ such that $f={\bf T}(U,\mu)$, where $U$ is the uniform distribution over the (bounded) domain of $\mu$. Our \Thm{GC-1651}, and its extension to general trade-off functions shows that Dong's construction is actually unique up to equivalence under refinement.
Su \cite{annurev:/content/journals/10.1146/annurev-statistics-112723-034158} also provides more insight about the connection to Blackwell's distinguishability by hypothesis testing \cite{af8b16cdacdb-3d4e-8fd0-40e54a9cc18a}. What is interesting here, is that Blackwell  \cite{af8b16cdacdb-3d4e-8fd0-40e54a9cc18a} suggests that it does not seem possible to extend hypothesis testing using hockey-stick gain functions to more complex privacy scenarios that require the analysis of many secrets \cite{DBLP:conf/icalp/AlvimACP11,DBLP:journals/jcs/AlvimACDP15}. This means that channels ${\mathbb C}_2$ corresponding to a trade-off function represents a worst-case scenario generating the tight $f$-DP constraint.

Awan et al. \cite{awan2023canonicalnoisedistribuionsprivate} have looked at sampling mechanisms for implementing exact $f$-DP constraints, especially for symmetric trade-off functions. Their proposal for efficient sampling for general $f_{\epsilon, \delta}$, for example could potentially be extended via our canonical channel representation for symmetric $f$-DP mechanisms (\Thm{Can-rep-1220}).   

Privacy purification was proposed by Lin et al.~\cite{lin2025purifyingapproximatediferentialprivacy}; our \Alg{Purify-1511} is based on their algorithm, using Geometric noise rather than Laplace noise as the final perturbation. Our lower bound should also apply to Laplace perturbation by our \Thm{Refinements-1523} here, since the Laplace perturber refines the Geometric perturber \cite{10.1109LICS52264.2021.9470718}. 
Sub-sampling has also been investigated by Balle et al.\cite{Balle_Barthe_Gaboardi_2020} wrt.\ an $(\epsilon,\delta)$-DP definition and Wang et al. \cite{pmlr-v89-wang19b} for Renyi-DP. Our analysis provides a full privacy profile, similar to the analysis in Dong et al.~\cite{10.1111/rssb.12454}.

Quantitative Information Flow for analysing security leaks in programs was first proposed by Clark et al.~\cite{DBLP:journals/jcs/ClarkHM07,ClarkHM01} Malacaria \cite{Malacaria:10}, and Smith \cite{Smith:2009aa}; the g-leakage framework was introduced by Alvim et al.~\cite{Alvim:2012aa} for general quantitative information flow analysis, and McIver et al.~\cite{McIver:10} for programs. Many of these ideas have been extended and developed by Alvim et al.~
\cite{Alvim:20a}.
A quantitative information flow semantics for programs was given by Gibbons et al.~\cite{Gibbons_McIver_Morgan_Schrijvers_2020}. Other verification techniques for privacy have been given in \cite{10.1145/2933575.2934554}, \cite{10.1007/978-3-030-34175-6_1},\cite{10.1145/2492061} and \cite{6957126}; a light-weight language for the verification of traditional differential privacy has been developed by Zhang and Kifer \cite{10.1145/3093333.3009884}. These methods use the ``worst-case'' features of differential privacy, and the simplifications in the analyses that it brings.

\section{Discussion and Future work}
We have demonstrated the equivalence between distinguishability via hypothesis testing and the QIF operational model for an information-aware program semantics. This enables novel composition rules for $f$-DP that can be used in the analysis for implementations for privacy-enhancing mechanisms. The QIF model in terms of information channels is particularly effective for monotone operators (such as $\parallel, \leftindex_p\oplus$ and pre-processing). We have also provided new analyses of privacy purification and sub-sampling. 

In future work we would like to make this analysis easier to apply directly to the verification of programs, for example by combining extant tools for analysing programs for information leaks~\cite{Gibbons_McIver_Morgan_Schrijvers_2020} with our \Alg{channel-alg-1734} to create optimal channel specifications of inputs.

\newpage 
\bibliographystyle{plain}
\bibliography{bibliography}

\newpage 
\appendix
\section{Proofs for \Sec{f-DP-QIF-Sec}}

\subsection{Proof for \Thm{QIF+HT-1213}}

\begin{enumerate}
    \item Fix $\alpha$. We set ${\cal T}(C)(\alpha)=p\leq q= {\cal T}(M)(\alpha)$. Wlog we assume that $p{\leq 1{-}p}$,  and $q{\leq 1{-}q}$, where
    \[
   C^\alpha = \begin{bmatrix}
    1{-}\alpha & \alpha\\
    p & 1{-}p
    \end{bmatrix}~~,~~
    M^\alpha = \begin{bmatrix}
    1{-}\alpha & \alpha\\
    q & 1{-}q
    \end{bmatrix} 
    \]
 We observe that
 \[
 \begin{bmatrix}
    1{-}\alpha & \alpha\\
    p & 1{-}p
    \end{bmatrix}
    \begin{bmatrix}
    a & 1{-}a\\
    b & 1{-}b
    \end{bmatrix} = \begin{bmatrix}
    1{-}\alpha & \alpha\\
    q & 1{-}q
    \end{bmatrix}~,
 \]
 where $b=(q{-}p)/(1{-}\alpha)(1{-}p)-p\alpha)$, and $a=((1{-}\alpha)-\alpha b)/\alpha$. The conditions above imply that $0\leq a, b\leq 1$ so that the refinement witness above is well-defined. 

\item Assume that $M^\alpha \sqsubset C^\alpha$. This means that with $p={\cal T}(M)(\alpha)$, and $q={\cal M}(C)(\alpha)$, we may depict $p, 1{-}p$ as the orange points on the horizontal in \Fig{HS-ref-805} and $q, 1{-}q$ as the blue points. From this we see that for any $h$ with $h < (1{-}p)/\alpha$, we must have by the construction that $V_{\underline{h}}[u\triangleright M]> V_{\underline{h}}[u\triangleright C]$, as shown. Now all we do is pick $h$ to be maximal such that $\textit{err}_C(h)\leq \alpha$. If we have equality, then we have immediately that
\[
V_{\underline{h}}[u \triangleright C]= V_{\underline{h}}[u \triangleright C^\alpha] < V_{\underline{h}}[u \triangleright M^\alpha] \leq V_{\underline{h}}[u \triangleright M]~,
\]
with the first equality following from \Def{err-1828}, and the last inequality since $M\sqsubseteq M^\alpha$.  
If $\textit{err}_C(h) {<} \alpha$, this can only happen if there is some $i^*$ such that $\sum_{i>i^*}p_i {<} \alpha {<} \sum_{i\geq i^*}p_i$. In this case we set $h=q_{i^*}/p_{i^*}$, and argue as above. 

\item Suppose that for all $h$, we have $V_{\underline{h}}[u \triangleright C] \leq V_{\underline{h}}[u \triangleright M]$; from this we deduce from (2) above that $M^\alpha \sqsubseteq C^\alpha$ for all $\alpha$. This implies that $M\sqsubseteq M^\alpha \sqsubseteq C^\alpha$, for all $\alpha$. This latter property implies that $M \sqsubseteq C$. This is because $M^\alpha \sqsubseteq C^\alpha$ for all $\alpha$ implies $M\sqsubseteq \min_\alpha C^\alpha = C$, which we prove in Appendix \Sec{App-B} is well defined and is equal to $C$.

The opposite direction follows immediately from refinement of channels.
\end{enumerate}

\section{Proofs from \Sec{Galois-sec}}\label{App-B}


\subsection{Proof for \Lem{wd-1650}}
Given any trade-off function $f\in {\mathbb F}$,  ${\cal C}(f)$ in \Def{glb-1642} is a well-defined channel in ${\mathbb C}_2$.

\begin{proof} 
The proof takes the following steps.
\begin{enumerate}
\item We show that $A \min B$ is well-defined for $A,B$ $2{\times}2$ channels. (Propositions~(\ref{Construction-AminR},\ref{General-min}).)

\item Next we provide a direct construction for $\min_{i\in {\cal I}} C_i$ is finite, and each $C_i$ is a $2\times 2$  (Proposition~(\ref{Finit MIN 1209}) and \Alg{MIN-1139});

\item We extend to the general case by appealing to the continuity properties of refinement relation $\sqsubseteq$~ (Prop.~\ref{GLB 1012}).
\end{enumerate}

\begin{proposition}\label{Construction-AminR}
Let channels $A$ and $R$ 
\[
A= \begin{bmatrix}
a & 1{-}a\\
b& 1{-}b
\end{bmatrix}
\quad \text{and}\quad
R= \begin{bmatrix}
1{-}r & r\\
1{-}s& s
\end{bmatrix}
\]
be incomparable under refinement. We assume wlog $b < a$ and $r < s$. Then
\[
A \min R = \begin{bmatrix}
  a & 1{-}a{-}r & r\\
  b & 1{-}b{-}s & s
\end{bmatrix}~
\]
is refined by both $A$ and $R$.
\end{proposition}
\begin{proof}
Since $N =2$, it suffices to focus on the second coordinate of the inners. For $A$ and $R$ to be incomparable under refinement we must have the inners of neither $A$ or $R$ is fully embedded in the convex hull of the other. We first consider the case where the centre of mass $m$ is in the middle and discuss later that the results hold under a general centre of mass. That is, we first consider \[\frac{b}{a+b} \leq \frac{1-s}{2-r-s} < m= \frac{1}{2} <\frac{1-b}{2-a-b} \leq \frac{s}{r+s}\]
From $\frac{1-b}{2-a-b} \leq \frac{s}{r+s}$ we conclude
\begin{align}
1-b &\leq \frac{s}{r+s}(2-a-b) \Rightarrow\\
1-b-s&\leq \frac{s}{r+s}(2-a-b)-s = \frac{s}{r+s}(2-a-b-r-s)
\end{align}
From $\frac{b}{a+b} \leq \frac{1-s}{2-r-s} $ we conclude
\begin{align}
1-s &\geq \frac{b}{a+b}(2-r-s) \Rightarrow\\
1-s-b&\geq \frac{b}{a+b}(2-r-s)-b = \frac{b}{a+b}(2-r-s-a-b)
\end{align}
Hence
\[\frac{b}{a+b}(2-r-s-a-b)\leq 1-b-s \leq \frac{s}{r+s}(2-a-b-r-s)\]
Since $\frac{b}{a+b} \leq \frac{s}{r+s}$, for the inequalities above to be true we must have $(2-a-b-r-s) \geq 0$, which implies $1-b-s \geq 0$ and $\frac{b}{a+b} \leq \frac{1-b-s}{2-r-s-a-b} \leq \frac{s}{r+s}$. Similarly, $1-a-r \geq 0$. Therefore, the channel 
\[
A \min R = \begin{bmatrix}
  a & 1{-}a{-}r & r\\
  b & 1{-}b{-}s & s
\end{bmatrix}~
\]
is well defined. It can be verified that pushing the uniform distribution into $A \min R$ will result in outers $\frac{1}{2}(a+b)$, $\frac{1}{2}(2-a-b-r-s)$ and $\frac{1}{2}(r+s)$ and give the centre of mass at $m = 1/2$.

This is shown in the figure below:
\bigskip

\begin{tikzpicture}[x=10cm, y=1cm]
  \draw[->] (0,0) -- (1.1,0) node[right] {$\mathbb{R}$};

  \coordinate (x0) at (0, 0);
  \coordinate (x1) at (0.2, 0);
  \coordinate (x2) at (0.35, 0);
  \coordinate (x3) at (0.5, 0);
  \coordinate (x4) at (0.6, 0);
    \coordinate (x42) at (0.75, 0);
  \coordinate (x5) at (0.9, 0);
    \coordinate (x6) at (1, 0);

  \foreach \x in {x0,x1,x2,x3,x4,x42,x5,x6} {
    \fill (\x) circle (0.01);
  }

  \node[below=4pt] at (x0) {$0$};
  \node[below=4pt] at (x1) {$\frac{b}{a+b}$};
  \node[below=4pt] at (x2) {$\frac{1-s}{2-r-s}$};
  \node[below=4pt] at (x3) {$\frac{1}{2}$};
  \node[below=4pt] at (x4) {$\frac{1-b}{2-a-b}$};
  \node[below=4pt] at (x42) {$\frac{1-b-s}{2-a-b-r-s}$};
  \node[below=4pt] at (x5) {$\frac{s}{r+s}$};
    \node[below=4pt] at (x6) {$1$};
\end{tikzpicture}
\bigskip

For the case of a general prior, the second coordinates of the inners on the barycentric line satisfy
\[\frac{\bar{\pi_1}b}{\pi_1a+\bar{\pi_1}b} \leq \frac{\bar{\pi_1}(1-s)}{\pi_1(1-r)+\bar{\pi_1}(1-s)} < \bar{\pi_1} < \frac{\bar{\pi_1}(1-b)}{\pi_1(1-a)+\bar{\pi_1}(1-b)} \leq \frac{\bar{\pi_1}s}{\pi_1r+\bar{\pi_1}s},\]
where $\bar{\pi_1} = 1-\pi_1$.
Using similar arguments as above, it can be verified that 
\[
A \min R = \begin{bmatrix}
  a ~~& 1{-}a{-}r ~~& r\\
  b ~~& 1{-}b{-}s ~~& s
\end{bmatrix}~
\]
will result in outers $\pi_1a+\bar{\pi_1}b$, $\pi_1(1-a-r)+\bar{\pi_1}(1-b-s)$ and $\pi_1r+\bar{\pi_1}s$ and give (second coordinate of) the centre of mass at $m = \bar{\pi_1}$.

Notice that this is the greatest lower bound, for if $M\sqsubseteq A$ and $M\sqsubseteq R$, then observe first that we can assume that $M$ has three columns. This is because for any $\alpha, \alpha'$ we can prove by direct calculation that $M \sqsubseteq M^\alpha \min M^{\alpha'}$, so we select $\alpha$ so that $M^\alpha\sqsubseteq A$ and $\alpha'$ such that $M^{\alpha'}\sqsubseteq R$. The result now follows since by the refinement assumption there must be a convex sum of the three columns of $M$ that are equal to $\begin{bmatrix} a\\ b\end{bmatrix}$, and similarly for $\begin{bmatrix} r\\ s\end{bmatrix}$. Since by the calculations above, the sum of those two columns is less than $\begin{bmatrix} 1\\ 1\end{bmatrix}$
this constitutes a well-defined refinement, which must be below $A\min R$, since the maximal refinement amongst two row channels of the form
\[
A \min R = \begin{bmatrix}
  a ~~& p \dots p' ~~& r\\
  b ~~& q \dots q' ~~& s
\end{bmatrix}~
\]
is $A \min R$, since the sum of the inner columns in the diagram is equal to $\begin{bmatrix}
  1-a - r\\
  1-b-s
\end{bmatrix}$.
\end{proof}

We now extend this fact to partial channels where all rows sum up to no more than 1.

\begin{proposition}\label{General-min}
Let partial channels $A$ and $R$ 
\[
A= \begin{bmatrix}
a & e{-}a\\
b& f{-}b
\end{bmatrix}
\quad \text{and}\quad
R= \begin{bmatrix}
e{-}r & r\\
f{-}s& s
\end{bmatrix}
\]
be incomparable under refinement. We assume wlog $b < a$ and $r < s$. Then
\[
A \min R = \begin{bmatrix}
  a ~~& e{-}a{-}r ~~& r\\
  b ~~& f{-}b{-}s ~~& s
\end{bmatrix}~
\]
is refined by both $A$ and $R$.
\end{proposition}
Note that even though the partial channels are incomparable, the corresponding rows in both channels sum up to the same value, $e$ and $f$, respectively. 
\begin{proof}
The proof follows  along the same lines as the previous proposition. We only consider uniform prior for simplicity. Since the channels are incomparable, we must have:

\[\frac{b}{a+b} \leq \frac{f-s}{e+f-r-s} < m= \frac{f}{2} <\frac{f-b}{e+f-a-b} \leq \frac{s}{r+s}\]
Hence
\[\frac{b}{a+b}(e+f-r-s-a-b)\leq f-b-s \leq \frac{s}{r+s}(e+f-a-b-r-s)\]
Since $\frac{b}{a+b} \leq \frac{s}{r+s}$, for the inequalities above to be true we must have $(e+f-a-b-r-s) \geq 0$, which implies $f-b-s \geq 0$ and $\frac{b}{a+b} \leq \frac{f-b-s}{e+f-r-s-a-b} \leq \frac{s}{r+s}$. Similarly, $e-a-r \geq 0$. Therefore, the partial channel 
\[
A \min R = \begin{bmatrix}
  a ~~& e{-}a{-}r ~~& r\\
  b ~~& f{-}b{-}s ~~& s
\end{bmatrix}~
\]
is well defined. 
\end{proof}

\end{proof}


\begin{definition}\label{d1732}
Given two $2\times 2$ channels $C,D$, where
\[
C= \begin{bmatrix}
a & a'\\
b & b'
\end{bmatrix}
~~~~D = \begin{bmatrix}
c & c'\\
d & d'
\end{bmatrix}
\]
where all rows sum to no more than $1$, and they have the same centroid, i.e.\  $a{+}a'= c{+}c'$, and $b{+}b'= d{+}d'$. 

For such (partial) channels We say that $C \lessdot D$ if
\[
\min\{a/(a+b), (a')/(a'{+}b') \} < \min\{c/(c+d), (c')/(c'{+}d') \}~.
\]
We write \textit{MinCol(C)}for the corresponding minimum column in $C$, i.e. 
\[
\textit{MinCol(C)} = \begin{bmatrix} a \\b\end{bmatrix}~~\textit{if} ~~\min\{a/(a+b), (a')/(a'{+}b') \}= a/(a+b)~,~\textit{otherwise}~~\begin{bmatrix} a' \\b'\end{bmatrix}~.
\] 
\end{definition}

\begin{proposition}\label{Finit MIN 1209}
Let $M_1\dots M_n$ be $2{\times}2$ channels. Then $\min_{1{\leq}i{\leq   n}}M_i$ is well-defined. 
\begin{proof}
We argue by induction on $n$, that the greatest lower bound of a set of $n$ $2{\times}2$ matrices all with the same centroid have collectively a greatest lower bound. If $n{=}2$ then the result follows from QIF fact 2.

Wlog assume that $M_1$ is the minimal matrix wrt.\ $\lessdot$ and that $\textit{MinCol}(M_1) = \begin{bmatrix} a \\b\end{bmatrix}$. Now form the partial channels
$M_1\min M_2, \dots M_1 \min M_n$, 
which are all well-defined. Notice that by construction 
$\textit{MinCol}(M_1 \min M_i)= \textit{MinCol}(M_1)$, so that we may assume, for each $i$
\[
M_1\min M_i = 
\begin{bmatrix}
a & 1{-}a{-} c'_i & c_i'\\
b & 1{-}b{-}d'_i & d'_i
\end{bmatrix}~,
\]
and in fact $\textit{MinCol}(M_1)$ 
becomes s column in the overall glb.

Now form the set of $n-1$ $2\times 2'$ matrices:
\[
M_i'= 
\begin{bmatrix}
 1{-}a{-} c'_i & c_i'\\
 1{-}b{-}d'_i & d'_i
\end{bmatrix}~~,~~~ 2\leq i \leq n~.
\]
Observe that the $M'_i$ all have the same centroid which is $v-\textit{MinCol}(M_1)$, where $v$ was the former centrid. Observe also that the $M'_i$ have two columns. Hence by the induction hypothesis the greatest lower bound $M'$ exists, with centroid $\begin{bmatrix} 1{-}a \\ 1{-}b \end{bmatrix}$.  Hence 
\[
\min_{1\leq i \leq n} M_i =
\begin{bmatrix}
\begin{bmatrix}
a \\
b
\end{bmatrix}~~ | ~~ M'
\end{bmatrix}
\]
\end{proof}
\end{proposition}

\begin{algorithm}
    \caption{Computing the Greatest lower bound}
    \label{MIN-1139}
    \begin{algorithmic}
 %
\State {\bf Assumes}  $M_1 \dots M_N$ are $2\times 2$ matrices with a common centroid.\\
\State {\bf input} $MM$\Comment{List of matrices}
 \State $n= 0$;
 \State  $G=\phi$
 \While{n>0} \Comment{{\bf invariant:} All matrices in $MM$ have a common centroid}
 \State $M$= smallest matrix in $MM$ wrt. $\lessdot$\Comment{Definition~\ref{d1732}}
 \State $JJ= [M\min M_i ~|~ 1{\leq}i{\leq}n,  M_i \neq M]$ \Comment{QIF fact 2}
 \State $G= G \cup MinCol(M)$ \Comment{QIF fact 2; becomes column for overall glb}
 \State $MM= [JJ[i]\setminus MinCol(M) ~|~ 1{\leq}i{\leq}n{-}1]$ \\ \Comment{Prop.~\ref{Finit MIN 1209}: channels have common centroid}
 \State $n++$
 \EndWhile
 \State {\bf Output} $G$
    \end{algorithmic}
\end{algorithm}

\begin{corollary}
Let $M\in \mathbb{\cal C}_2$ be finite. Then ${\cal C}({\cal T}(M))=M$.
\begin{proof}
Observe that ${\cal T}(M)$ is piecewise linear. This means that the specialisation of \Alg{MIN-1139} to \Alg{channel-alg-1734} shows directly that the same columns of $M$ are recovered.

We illustrate the simplification by demonstrating the first few steps of \Alg{MIN-1139}.

\[\label{M-chan}
M ~~=~~ \begin{bmatrix}
p_0 & p_1 & \dots & p_{n-1} & p_n\\
q_0 & q_1 & \dots & q_{n-1} & q_n
\end{bmatrix}~, ~~~\textit{where}~~~ q_i/p_i~\textit{increasing}
\]
We can write down the facet points, $\alpha_i$ and the corresponding ${\cal T}(M)(\alpha_i)$:

\[
\alpha_0= p_n,~ {\cal T}(M)(\alpha_0)=1{-}q_n,~~\alpha_1= p_n{+}p_{n-1},~ {\cal T}(M)(\alpha_1)=1{-}(q_n{+}q_{n-1})~\dots
\]

Thus,
\[
M^{\alpha_0} \min M^{\alpha_1} =
\begin{bmatrix} 1{-}(p_n{+}p_{n-1}) & p_{n-1} & p_n\\
1{-}(q_n{+}q_{n-1}) & q_{n-1} & q_n
\end{bmatrix}~.
\]

\[
 (M^{\alpha_0} \min M^{\alpha_2}) =\begin{bmatrix} 1{-}(p_n{+}p_{n-1}{+}p_{n-2}) & p_{n-1}{+}p_{n-2} & p_n\\
1{-}(q_n{+}q_{n-1}{+}q_{n-2}) & q_{n-1}{+}q_{n-2} & q_n\\
\end{bmatrix}
\]
we can now compute $M^{\alpha_0} \min M^{\alpha_2}$, as:
\[
 (M^{\alpha_0} \min M^{\alpha_2})\min (M^{\alpha_0} \min M^{\alpha_1})=
\begin{bmatrix} 1{-}(p_n{+}p_{n-1}{+}p_{n-2}) & p_{n-2} & p_{n-1} & p_n\\
1{-}(q_n{+}q_{n-1}{+}q_{n-2}) & q_{n-2}&q_{n-1} & q_n\\
\end{bmatrix}~.
\]
Thus \Alg{MIN-1139} computes $M$ again from the facet points.

Finally we observe that for any $\alpha$ that lies between two facet points $\alpha_k$ and $\alpha_{k+1}$, we must have that $M^{\alpha}$ refines $M^{\alpha_k}\min M^{\alpha_{k+1}}$, which follows since $\alpha= \alpha_k+\lambda_{n-k-1}$ for some $0<\lambda<1$, thus:

\[
M^\alpha = \begin{bmatrix} 1{-}\alpha_{k+1} {+} (1{-}\lambda)p_{n-k-1} & \alpha_k
+ \lambda p_{n-k-1}\\
{\cal T}(M)(\alpha_{k+1}) {+} (1{-}\lambda)q_{n-k-1} & 1- {\cal T}(M)(\alpha_k) 
+ \lambda q_{n-k-1}
\end{bmatrix}
\]
whereas, since $\alpha_{k+1}=\alpha_k+p_{{n-k-1}}$,
\[
M^{\alpha_k}\min M^{\alpha_{k+1}} =
\begin{bmatrix}
1{-}\alpha_{k+1} & p_{n-k-1} & \alpha_{k}\\
{\cal T}(M)(\alpha_{k+1}) & q_{n-k-1} & 1{-}{\cal T}(M)(\alpha_k)
\end{bmatrix}~.
\]
We can now use the refinement witness
\[
W= \begin{bmatrix}
1 & 0\\
1-\lambda & \lambda\\
0 & 1
\end{bmatrix}
\]
to show that $(M^{\alpha_k}\min M^{\alpha_{k+1}})\cdot W = M^{\alpha}$, as required.

\end{proof}
\end{corollary}

To generalise these results in ${\mathbb C}_2$, we must consider limits of directed sets within a compact set. We use abstract channels for this \cite{McIver:2014aa}, which is a generalisation of ${\mathbb C}_2$. The set of abstract channels is where we find the full correspondence between ${\mathbb F}$ and abstract channels. In the main paper we restrict to piecewise-linear trade-off functions, for ease of exposition.

\begin{proposition}\label{GLB 1012}
If ${\cal I}$ is an index set and $A_i$ is a channel with two columns then $\min_{i\in {\cal I}}A_i$ is well-defined.

\begin{proof}
$\min_{i\in {\cal I}} \Delta_i = \min_{{S}\subseteq {\cal I}} \min_{i\in S} \Delta_i$, where $S$ is any finite subset if ${\cal I}$. 

This follows since the set of $\min_{i\in S}\Delta_i$ for all finite $S$ exists by Prop.~\ref{Finit MIN 1209}  form a meet semi-Lattice, and refinement relation $\sqsubseteq$ is continuous for limits \cite{McIver:12}.
\end{proof}
\end{proposition}

\section{QIF Facts for computing $\min$}
From \cite{Alvim:20a} Lemma 12.2, we know that within a state space of size $N$, a hyper distribution with $N$ or fewer independent inners can be refined to any hyper lying within its convex closure that has the same centre of mass. For $N = 2$ secrets, this becomes very simple. The inners can be represented on the barycentric interval $[0,1]$, say by the probability they assign to the second secret. Let two independent inners of a hyper be represented by two distinct $p,q$ on the barycentric interval $[0,1]$ with the centre of mass $m$, such that $p < m < q$. They are refined by any hyper with the same centre of mass such that $p \leq p' < m < q'\leq q$. This relationship does not depend on $m$. It suffices to consider the special case where $m = 1/2$, where the centre of mass of inners is balanced. This can be achieved by pushing the uniform distribution through the $2\times 2$ channel.

\begin{proposition}
WLOG, assume that $b <a, b' < a'$. Then 
\[
A = \begin{bmatrix}
a & 1{-}a\\
b & 1{-}b
\end{bmatrix} \sqsubseteq
\begin{bmatrix}
a' & 1{-}a'\\
b' & 1{-}b'
\end{bmatrix} = A' ~,~~~\textit{iff}~~~~ \frac{a'}{b'} \leq \frac{a}{b}~~\textit{and}~~\frac{1-a'}{1-b'} \geq \frac{1-a}{1-b}~.
\]
\end{proposition}

\begin{proof}
We push uniform input distribution $\pi = [1/2, 1/2]^T$ through $A$ and $A'$. In the barycentric representation, this results in the second coordinate of the inners for $A$ to be $p = \frac{b}{a+b}$ and $q = \frac{1-b}{2-a-b}$ and for those of the inners of $A'$ to be $p' = \frac{b'}{a'+b'}$ and $q' = \frac{1-b'}{2-a'-b'}$. For $p \leq p' < \frac{1}{2} < q'\leq q$ we must have
\[
\frac{b}{a+b} \leq \frac{b'}{a'+b'}\quad\Leftrightarrow\quad \frac{a'}{b'} \leq \frac{a}{b}\]
and
\[
 \frac{1-b}{2-a-b} \geq \frac{1-b'}{2-a'-b'} \quad\Leftrightarrow\quad \frac{1-a'}{1-b'} \geq\frac{1-a}{1-b}
\]
This is shown in the figure below:
\bigskip

\begin{tikzpicture}[x=10cm, y=1cm]
  \draw[->] (0,0) -- (1.1,0) node[right] {$\mathbb{R}$};

  \coordinate (x0) at (0, 0);
  \coordinate (x1) at (0.2, 0);
  \coordinate (x2) at (0.35, 0);
  \coordinate (x3) at (0.5, 0);
  \coordinate (x4) at (0.65, 0);
  \coordinate (x5) at (0.85, 0);
    \coordinate (x6) at (1, 0);

  \foreach \x in {x0,x1,x2,x3,x4,x5,x6} {
    \fill (\x) circle (0.01);
  }

  \node[below=4pt] at (x0) {$0$};
  \node[below=4pt] at (x1) {$\frac{b}{a+b}$};
  \node[below=4pt] at (x2) {$\frac{b'}{a'+b'}$};
  \node[below=4pt] at (x3) {$\frac{1}{2}$};
  \node[below=4pt] at (x4) {$\frac{1-b'}{2-a'-b'}$};
  \node[below=4pt] at (x5) {$\frac{1-b}{2-a-b}$};
    \node[below=4pt] at (x6) {$1$};
\end{tikzpicture}
\bigskip
\end{proof}
We remark the relation is purely a relationship between channel matrix values and is independent of the prior that was pushed through. It is easy to verify that even if the inners are generated by a general prior $\pi = [\pi_1, 1-\pi_1]$, the conditions in the proposition guarantee that $p=\frac{(1-\pi_1)b}{\pi_1a+(1-\pi_1)b} \leq p'=\frac{(1-\pi_1)b'}{\pi_1a'+(1-\pi_1)b'} < m = 1-\pi_1 < \frac{(1-\pi_1)(1-b')}{\pi_1(1-a')+(1-\pi_1)(1-b')} =q' \leq \frac{(1-\pi_1)(1-b)}{\pi_1(1-a)+(1-\pi_1)(1-b)}=q$.

\newpage


\section{Illustration of privacy profiles for compositions}

\begin{figure}
\centering
\includegraphics[width=0.6\textwidth]{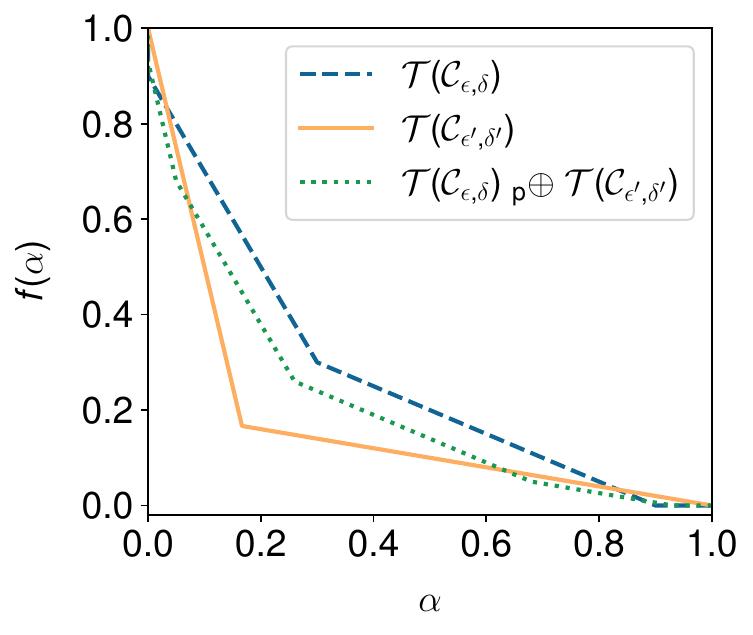}
\caption{Graph illustrating visible choice.}\label{Vis choice}
\end{figure}

\section{Detailed calculations for \Sec{Applications-main}}

\subsection{Analysis of Purification, \Alg{Purify-1511}}
We now provide a privacy analysis of \Alg{Purify-1511}. Recall that the two main operations of \Alg{Purify-1511} are 1) a hidden probabilistic choice between the true mechanism output and a uniform sample on $\mathcal{Y}$, followed by adding a sample from  Geometric distribution $G_{\epsilon'}$. That is,
\begin{align}\label{eq:purificationpart1}
(y \leftarrow M(x)) \leftindex_r\boxplus (y \leftarrow U(\mathcal{Y})) ~;~ z \leftarrow y + G_{\epsilon'}(0) 
\end{align}

We now study two different scenarios. The first one is where the output range of $M$ is the same as that of the uniform distribution, $\mathcal{Y}$. The second scenario is where the output of $M$ is not fully covered by $\mathcal{Y}$. In many practical cases this happens when the mechanism output is continuous over an unbounded region. 

Since according to \Cor{cor:eps:delta:refined:by:M} for any $(\epsilon,\delta)$-DP mechanism $\mathcal{M}$, we have $C_{\epsilon, \delta} = {\cal C}(f_{\epsilon,\delta}) \sqsubseteq M$, we study these two scenarios on $C_{\epsilon, \delta}$. 

 Consider the hidden choice channel
\[D = C_{\epsilon, \delta}~\leftindex_r\boxplus ~U(\{1,2,3,4\})
\]

It can be verified that under the condition
\[
e^{\epsilon} > \frac{1}{1-2\delta}+\frac{(1-r)(1-\delta)\delta}{0.25r}
\]
$D$ can be written in increasing order of the ratio $q/p$ as follows where $a = (1-\delta)/(1+e^{\epsilon})$:
\[
D ~~=~~ 
\begin{bmatrix}
 0.25r {+} (1{-}r)ae^{\epsilon} ~~ &~~ 0.25r{+}(1{-}r)\delta~~& ~~0.25r ~~& 0.25r {+} (1{-}r)a \\
~~ 0.25r {+} (1{-}r)a ~~&~~0.25r ~~& ~~ 0.25r{+}(1{-}r)\delta~~& 0.25r {+} (1{-}r)ae^{\epsilon} 
\end{bmatrix}.
\]
We first observe that the effective $\delta' = 0 $. That is, for the case where the range of $\mathcal{M}$ is also $\mathcal{Y}$, hidden choice with uniform mechanism is sufficient to purify $C_{\epsilon, \delta}$.

Let $r = 0.25$, $\delta = 0.1$ and $e^\epsilon = 3$. We get
\begin{align}\label{D:124}
D ~~=~~ 
\begin{bmatrix}
 0.56875 ~~ &~~ 0.1375~~& ~~0.0625 ~~& 0.23125 \\
0.23125 ~~&0.0625 ~~& ~~ 0.1375~~& 0.56875 
\end{bmatrix}.
\end{align}
The $f$-differential privacy profile of $C_{\epsilon, \delta}$ and $D$ are shown in the figure below, verifying higher privacy of $D$ and purification of $C_{\epsilon, \delta}$.

\begin{tikzpicture}
  \begin{axis}[
    width=9cm,
    axis lines=middle,
    axis line style={-}, 
    xmin=-0.05, xmax=1.05,
    ymin=-0.05, ymax=1.05,
    xtick={0,0.25,0.5,0.75,1},
    ytick={0,0.25,0.5,0.75,1},
    grid=both,
    xlabel={$\alpha$},
    ylabel={$\beta$},
    xlabel style={
      at={(axis description cs:0.5,0)},
      anchor=north,
      yshift=-6pt
    },
    ylabel style={
      at={(axis description cs:0,0.5)},
      anchor=east,
      xshift=-8pt
    },
    enlargelimits=false,
    clip=false,
    legend style={
      at={(0.97,0.97)}, 
      anchor=north east,
      draw=none,
      fill=none
    }
  ]

      \addplot[thick, mark=*, mark options={fill=white}] coordinates {
      (0,0.9)
      (0.225,0.225)
      (0.9,0)
    };
    \addlegendentry{original $C_{\epsilon,\delta}$}
    \addplot[ thick, dashed, mark=*] coordinates {
      (0,1)
      (0.23125,0.43125)
      (0.29375,0.29375)
      (0.43125,0.23125)
      (1,0)
    };
    \addlegendentry{Hidden choice $D$}


 
    \addplot[semithick, dashdotted, mark=triangle*, mark options={fill=white}] coordinates {
      (0,1)
      (0.3333,0.3333)
      (1,0)
    };\addlegendentry{random response with $\log 2$}
    \addplot[semithick, dotted, mark=triangle*, mark options={fill=white}] coordinates {
      (0,1)
      (0.2453,0.5453)
      (0.3750,0.3750)
      (0.5453,0.2453)
      (1,0)
    };
    \addlegendentry{$D$ Post-processed as $DG$}

        \addplot[semithick, dashdotted, mark=triangle*, mark options={fill=white}, color = red] coordinates {
      (0,1)
      (0.1958,0.6708)
      (0.3916,0.3916)
      (0.6708,0.1958)
      (1,0)
    };
    \addlegendentry{$C_{\epsilon, \delta}G$}
  \end{axis}
\end{tikzpicture}

We next perform addition of a truncated Geometric mechanism $G_{\epsilon'}$ to the output of the hidden choice procedure. This can be achieved by post-processing $D$ as $Z = DG_{\epsilon'}$, where $G_{\epsilon'}$ is the matrix for a truncated Geometric mechanism on 4 secrets. Choosing  $\epsilon' = \log 2$, we get\footnote{$G_{\epsilon'}$ is specified below.}
\[
Z ~~=~~ 
\begin{bmatrix}
 0.4547 ~~ &~~ 0.1703~~& ~~0.1297 ~~& 0.2453 \\
0.2453 ~~&0.1297 ~~& ~~ 0.1703~~& 0.4547 
\end{bmatrix}.
\]
This is plotted in the same figure. For comparison, we have also plotted the $f$-differential privacy profile of a random response mechanism with $\epsilon' = \log_2$ and that of $C_{\epsilon, \delta}G$. It is observed that the output of \Alg{Purify-1511} is more private than the corresponding random response with the same $\epsilon'$ as the Geometric mechanism, but less private than $C_{\epsilon, \delta}G$.

For the second scenario, let us consider a visible probabilistic choice between the original $C_{\epsilon, \delta_1}$ and a mechanism that reveals fully the secrets with probability $\delta_2$. That is, 
\[
C'_{\epsilon, \delta} = \begin{bmatrix} 1&0\\0&1
\end{bmatrix}\leftindex_{\delta_2} \oplus C_{\epsilon, \delta_1}
\]
where $\delta = \delta_2+(1-\delta_2)\delta_1$.
This can represent the case where the range of $M$ is unbounded and hence some '$\delta$' terms remain unaddressed by the hidden choice between $M$ and the uniform mechanism over a bounded region. We get
\begin{align}
D'&= \left(\begin{bmatrix} 1&0\\0&1
\end{bmatrix}\leftindex_{\delta_2} \oplus C_{\epsilon, \delta_1}\right)\leftindex_r' \boxplus ~U(\{3,4,5,6\}) 
\\&=
\begin{bmatrix} 1&0\\0&1
\end{bmatrix} \leftindex_{r'\delta_2} \oplus ( C_{\epsilon, \delta_1}\quad\leftindex_{\frac{r'(1-\delta_2)}{r'(1-\delta_2)+(1-r')}} \boxplus \quad U(\{3,4,5,6\})
\end{align}
where the last equality follows from \Thm{composition-results-1310} part 5. Let us set $\delta_2 = 0.05$, $\delta_1 = 0.1$, $\epsilon = \log 3$ and $r' = 0.25974$. These values ensure that $\frac{r'(1-\delta_2)}{r'(1-\delta_2)+(1-r')} = 0.25$ and hence,
\[
D'= \begin{bmatrix} 1&0\\0&1
\end{bmatrix}\leftindex_{\delta_2} \oplus D
\]
where $D$ was given in \eqref{D:124}.
Post-processing $D'$ with a truncated Geometric mechanism on 6 secrets and $\epsilon' = \log 2$ gives \footnote{Matrix $G_{\epsilon'}$ is specified below.}
\[
Z' ~~=~~ D' G_{\epsilon'} =  
\begin{bmatrix}
 0.2493 ~~ &~~ 0.2243~~& ~~0.1660 ~~& 0.1253~~&~~0.1176~~&~~0.1176 \\
0.1176 ~~&0.1176 ~~& ~~ 0.1253~~& 0.1660~~&~~0.2242~~&0.2493 
\end{bmatrix}.
\]

\begin{tikzpicture}
  \begin{axis}[
    width=9cm,
    axis lines=middle,
    axis line style={-}, 
    xmin=-0.05, xmax=1.05,
    ymin=-0.05, ymax=1.05,
    xtick={0,0.25,0.5,0.75,1},
    ytick={0,0.25,0.5,0.75,1},
    grid=both,
    xlabel={$\alpha$},
    ylabel={$\beta$},
    xlabel style={
      at={(axis description cs:0.5,0)},
      anchor=north,
      yshift=-6pt
    },
    ylabel style={
      at={(axis description cs:0,0.5)},
      anchor=east,
      xshift=-8pt
    },
    enlargelimits=false,
    clip=false,
    legend style={
      at={(0.97,0.97)}, 
      anchor=north east,
      draw=none,
      fill=none
    }
  ]

      \addplot[thick, mark=*, mark options={fill=white}] coordinates {
      (0,0.855)
      (0.2137,0.2137)
      (0.855,0)
    };
    \addlegendentry{original $C_{\epsilon,\delta}$}
    \addplot[ thick, dashed, mark=*] coordinates {
      (0,0.95)
      (0.2197,0.4097)
      (0.2791,0.2791)
      (0.4097,0.2197)
      (0.95,0)
    };
    \addlegendentry{$D'$: non-covering $U$}
    \addplot[semithick, dashdotted, mark=triangle*, mark options={fill=white}] coordinates {
      (0,1)
      (0.3333,0.3333)
      (1,0)
    };\addlegendentry{random response with $\log 2$}
    \addplot[semithick, dotted, mark=triangle*, mark options={fill=white}] coordinates {
      (0,1)
      (0.1176,0.7507)
      (0.2352,0.5265)
      (0.3605,0.3605)
      (0.5265,0.2352)
      (0.7507,0.1176)
      (1,0)
    };
    \addlegendentry{Post-processing $Z' = D'G$}

  \end{axis}
\end{tikzpicture}

\subsubsection{Definition of Geometric post-processors}
The following matrices specify geometric mechanism for $\epsilon' = 
log 2$ on 4 and 6 secrets, respectively.
\[
G_{\epsilon'} ~~=~~ 
\begin{bmatrix}
 \frac{2}{3} ~~ &~~ \frac{1}{6}~~& ~~\frac{1}{12} ~~& \frac{1}{12} \\\\
\frac{1}{3} ~~ &~~ \frac{1}{3}~~& ~~\frac{1}{6} ~~& \frac{1}{6} \\\\
\frac{1}{6} ~~ &~~ \frac{1}{6}~~& ~~\frac{1}{3} ~~& \frac{1}{3}\\\\
\frac{1}{12} ~~ &~~ \frac{1}{12}~~& ~~\frac{1}{6} ~~& \frac{2}{3}
\end{bmatrix}.
\]

\[
G_{\epsilon'} ~~=~~ 
\begin{bmatrix}
 \frac{2}{3} ~~ &~~ \frac{1}{6}~~& ~~\frac{1}{12} ~~& \frac{1}{24} ~~&\frac{1}{48}~~&\frac{1}{48}\\\\
\frac{1}{3} ~~ &~~ \frac{1}{3}~~& ~~\frac{1}{6} ~~& \frac{1}{12} ~~& \frac{1}{24} ~~& \frac{1}{24}\\\\
\frac{1}{6} ~~ &~~ \frac{1}{6}~~& ~~\frac{1}{3} ~~& \frac{1}{6}~~& \frac{1}{12}~~& \frac{1}{12}\\\\
\frac{1}{12} ~~ &~~ \frac{1}{12}~~& ~~\frac{1}{6} ~~& \frac{1}{3}~~& \frac{1}{6}~~& \frac{1}{6}\\\\
\frac{1}{24} ~~ &~~ \frac{1}{24}~~& ~~\frac{1}{12} ~~& \frac{1}{6}~~& \frac{1}{3}~~& \frac{1}{3}\\\\
\frac{1}{48} ~~ &~~ \frac{1}{48}~~& ~~\frac{1}{24} ~~& \frac{1}{12}~~& \frac{1}{6}~~& \frac{2}{3}
\end{bmatrix}.
\]

\subsection{Sub-sampling}

Here, we detail the calculations for the analysis of sub-sampling.

Let us assume dataset $D$ has $n$ elements denoted by $m_{1},\cdots, m_n$ and dataset $D'$ has exactly the same elements as $D$ plus one more element denoted by $m_{n+1}$. We can characterise the input-output process of Poisson sub-sampling where each element is independently sampled with probability $\gamma$ as follows. If the input is a subset of $D$, say $X$, the output is the same subset $X$  with probability $1-\gamma$ and $X \cup \{m_{n+1}\}$ with probability $\gamma$.

\begin{theorem}
For Poisson sub-sampling the pre-processing matrix is given by
\[
P = \begin{bmatrix}
1 & 0\\
1-\gamma & \gamma
\end{bmatrix}
%
\]
\end{theorem}

Because we have shown that pro-processing preserves refinement, the above result can be used to argue that we can understand how the privacy profile changes of a mechanism changes by performing the sub-sampling on the worst $(\epsilon,\delta)$-DP channel $C_{\epsilon, \delta}$.

\begin{example}
Now let us apply this pre-processing to the worst-case $(\epsilon,\delta)$-DP channel there is, $C_{\epsilon,\delta}$. We have

\begin{align*}
P. C_{\epsilon,\delta} = \begin{bmatrix}
 \delta &a e^{\epsilon} &a & 0\\(1-\gamma)\delta\quad&
 (1-\gamma)a e^{\epsilon}+\gamma a\quad&(1-\gamma)a+\gamma a e^{\epsilon}\quad&\gamma\delta
\end{bmatrix}
\end{align*}
where $a = (1-\delta)/(1+e^{\epsilon})$. 
We compute $(\alpha,\beta)$ pairs at the following facet points

\begin{itemize}
\item \textbf{Point $A^s$:} $\alpha = 0$, $f(\alpha) = 1-\gamma\delta$.
\item \textbf{Point $B^s$:} $\alpha = a$, $f(\alpha) = (1-\gamma)\delta+(1-\gamma)ae^{\epsilon}+\gamma a $
\item \textbf{Point $C^s$:} $\alpha = a + a e^\epsilon = (1-\delta)$, $f(\alpha) = (1-\gamma)\delta$

\item \textbf{Point $D^s$:} $\alpha = 1$, $f(\alpha) = 0$
\end{itemize}

If we compare this with the critical points of the original $C_{\epsilon,\delta}$ mechanism, do we see the improvement in privacy profile (larger $\beta$ for same $\alpha$)?

\begin{itemize}
\item $\alpha = 0$, $f(\alpha) = 1-\delta < 1-\gamma\delta$
\item $\alpha = a = (1-\delta)/(1+e^{\epsilon})$, $f(\alpha) = a = (1-\delta)/(1+e^{\epsilon})$
\item $\alpha =  (1-\delta)$,  $f(\alpha) = 0 < (1-\gamma)\delta$
\end{itemize}
Therefore, the privacy profile of sampling is better than $C_{\epsilon,\delta}$ at extreme points $\alpha = 0$ and $\alpha = 1-\delta$ and $\alpha = 1$. It remains to show that at the remaining critical value $\alpha = a$, we have  
\begin{align*}
(1-\gamma)\delta+(1-\gamma)ae^{\epsilon}+\gamma a  -a =(1-\gamma)\delta + a (1-\gamma)(e^{\epsilon}-1) \geq 0
\end{align*}
Therefore, Poisson sampling improves $f$-privacy as expected. The following example visualizes this for $\epsilon = 1$, $\delta = 0.1$ and $\gamma = 0.2$.

\begin{tikzpicture}
  \begin{axis}[
    width=10cm,
    axis equal image,
    axis lines=left,
    xlabel={$\alpha$},
    ylabel={$f(\alpha)$},
    xmin=0, xmax=1,
    ymin=0, ymax=1,
    grid=both,
    legend pos=north east,
    xtick={0,0.25,0.5,0.75,1},
    ytick={0,0.25,0.5,0.75,1}
  ]

  \addplot[
    color=blue,
    thick,
    mark=*,
    mark options={fill=blue}
  ]
  coordinates {
    (0, 0.9)
    (0.24198, 0.24198)
    (0.9, 0)
  };
  \addlegendentry{$C_{\epsilon,\delta}$, $\epsilon = 1$, $\delta = 0.1$.}

  \addplot[
    color=red,
    thick,
    mark=square*,
    mark options={fill=red}
  ]
  coordinates {
    (0, 0.98)
    (0.24198, 0.654)
    (0.9, 0.08)
    (1, 0)
  };
  \addlegendentry{Sampled $C_{\epsilon,\delta}$, $\gamma = 0.2$}

  \end{axis}
\end{tikzpicture}

We note that the privacy profile of the sub-sampled mechanism is not symmetric. Symmetry comes from considering that the roles of $D,D'$ above is symmetric and in practice, it could that be $D$ is the dataset with extra element $m_{n+1}$. 
\end{example}

\end{document}

%% file: macros.tex
\newcommand\Dist {\mathbb{D}}

\newcommand\Sec[1] {\S\ref{#1}}
\newcommand\Eqn[1] {Eqn~(\ref{#1})}
\newcommand\Thm[1] {Thm~\ref{#1}}
\newcommand\Lem[1] {Lem.~\ref{#1}}
\newcommand\Def[1] {Defn~\ref{#1}}
\newcommand\Cor[1] {Cor.~\ref{#1}}
\newcommand\Fig[1] {Fig.~\ref{#1}}
\newcommand\Alg[1] {Alg.~\ref{#1}}

\newcommand\calx {\mathcal{X}}
\newcommand\caly {\mathcal{Y}}
\newcommand\calz {\mathcal{Z}}

\newif\ifqif

\qiftrue

\ifqif
    \newcommand\px      {\pi_x}    
    
    \newcommand\CondYX  {C_{x,y}}
    
\else
    \newcommand\px   {P_X(x)}

    \newcommand\CondYX  {P_{Y|X}(y|x)}
    
\fi

%% file: fossacs_ff.bbl
\begin{thebibliography}{10}

\bibitem{Alvim:20a}
M.~Alvim, K.~Chatzikokolakis, A.K. McIver, C.C. Morgan, G.S. Smith, and
  C.~Palamidessi.
\newblock {\em The Science of Quantitative Information Flow}.
\newblock Information Security and Cryptography. 2020.

\bibitem{Alvim:2012aa}
M.~S. Alvim, K.~Chatzikokolakis, C.~Palamidessi, and G.S. Smith.
\newblock Measuring information leakage using generalized gain functions.
\newblock In {\em Proc.\ 25th IEEE Computer Security Foundations Symposium (CSF
  2012)}, pages 265--279, June 2012.

\bibitem{DBLP:journals/jcs/AlvimACDP15}
M{\'{a}}rio~S. Alvim, Miguel~E. Andr{\'{e}}s, Konstantinos Chatzikokolakis,
  Pierpaolo Degano, and Catuscia Palamidessi.
\newblock On the information leakage of differentially-private mechanisms.
\newblock {\em Journal of Computer Security}, 23(4):427--469, 2015.

\bibitem{DBLP:conf/icalp/AlvimACP11}
M{\'{a}}rio~S. Alvim, Miguel~E. Andr{\'{e}}s, Konstantinos Chatzikokolakis, and
  Catuscia Palamidessi.
\newblock On the relation between differential privacy and quantitative
  information flow.
\newblock In {\em Automata, Languages and Programming - 38th International
  Colloquium, {ICALP} 2011, Zurich, Switzerland, July 4-8, 2011, Proceedings,
  Part {II}}, pages 60--76, 2011.

\bibitem{awan2023canonicalnoisedistribuionsprivate}
Jordan Awan and Salil Vadhan.
\newblock Canonical noise distributions and private hypothesis tests, 2023.

\bibitem{Balle_Barthe_Gaboardi_2020}
Borja Balle, Gilles Barthe, and Marco Gaboardi.
\newblock Privacy profiles and amplification by subsampling.
\newblock 10, Jan. 2020.

\bibitem{6957126}
Gilles Barthe, Marco Gaboardi, Emilio~Jesus Gallego~Arias, Justin Hsu, Cesar
  Kunz, and Pierre-Yves Strub.
\newblock { Proving Differential Privacy in Hoare Logic }.
\newblock In {\em 2014 IEEE 27th Computer Security Foundations Symposium
  (CSF)}, pages 411--424. IEEE Computer Society, 2014.

\bibitem{10.1145/2933575.2934554}
Gilles Barthe, Marco Gaboardi, Benjamin Gr\'{e}goire, Justin Hsu, and
  Pierre-Yves Strub.
\newblock Proving differential privacy via probabilistic couplings.
\newblock LICS '16, page 749–758. Association for Computing Machinery, 2016.

\bibitem{10.1145/2492061}
Gilles Barthe, Boris K\"{o}pf, Federico Olmedo, and Santiago
  Zanella-B\'{e}guelin.
\newblock Probabilistic relational reasoning for differential privacy.
\newblock {\em ACM Trans. Program. Lang. Syst.}, 35(3), November 2013.

\bibitem{af8b16cdacdb-3d4e-8fd0-40e54a9cc18a}
David Blackwell.
\newblock Equivalent comparisons of experiments.
\newblock {\em The Annals of Mathematical Statistics}, 24(2):265--272, 1953.

\bibitem{ClarkHM01}
David Clark, Sebastian Hunt, and Pasquale Malacaria.
\newblock Quantitative analysis of the leakage of confidential data.
\newblock {\em Electr. Notes Theor. Comput. Sci.}, 59(3):238--251, 2001.

\bibitem{DBLP:journals/jcs/ClarkHM07}
David Clark, Sebastian Hunt, and Pasquale Malacaria.
\newblock A static analysis for quantifying information flow in a simple
  imperative language.
\newblock {\em J. Comput. Secur.}, 15(3):321--371, 2007.

\bibitem{10.1111/rssb.12454}
Jinshuo Dong, Aaron Roth, and Weijie~J. Su.
\newblock Gaussian differential privacy.
\newblock {\em Journal of the Royal Statistical Society Series B: Statistical
  Methodology}, 84(1):3--37, 02 2022.

\bibitem{Dwork:2006aa}
Cynthia Dwork.
\newblock Differential privacy.
\newblock In {\em Proc.\ 33rd International Colloquium on Automata, Languages,
  and Programming (ICALP 2006)}, pages 1--12, 2006.

\bibitem{10.1109LICS52264.2021.9470718}
Natasha Fernandes, Annabelle McIver, and Carroll Morgan.
\newblock The laplace mechanism has optimal utility for differential privacy
  over continuous queries.
\newblock In {\em Proceedings of the 36th Annual ACM/IEEE Symposium on Logic in
  Computer Science}, LICS '21. IEEE Press, 2021.

\bibitem{Gibbons_McIver_Morgan_Schrijvers_2020}
Jeremy Gibbons, Annabelle McIver, Carroll Morgan, and Tom Schrijvers.
\newblock {\em Quantitative Information Flow with Monads in Haskell}, page
  391–448.
\newblock Cambridge University Press, 2020.

\bibitem{NPL1933}
Neyman Jerzy and Pearson~Egon Sharpe.
\newblock On the problem of the most efficient tests of statistical hypotheses.
\newblock {\em Philosophical Transactions of the Royal Society of London},
  1933.

\bibitem{lin2025purifyingapproximatediferentialprivacy}
Yingyu Lin, Erchi Wang, Yi-An Ma, and Yu-Xiang Wang.
\newblock Purifying approximate differential privacy with randomized
  post-processing, 2025.

\bibitem{Malacaria:10}
P.~Malacaria.
\newblock Risk assessment of security threats for looping constructs.
\newblock {\em Journal of Computer Security}, 18(2):191--228, 2010.

\bibitem{McIver:10}
Annabelle McIver, Larissa Meinicke, and Carroll Morgan.
\newblock Compositional closure for {Bayes Risk} in probabilistic
  noninterference.
\newblock In {\em Proceedings of the 37th international colloquium conference
  on Automata, languages and programming: Part II}, volume 6199 of {\em
  ICALP'10}, pages 223--235, Berlin, Heidelberg, 2010.

\bibitem{McIver:12}
Annabelle McIver, Larissa Meinicke, and Carroll Morgan.
\newblock A {Kantorovich}-monadic powerdomain for information hiding, with
  probability and nondeterminism.
\newblock In {\em Proc. {LiCS} 2012}, 2012.

\bibitem{10.1007/978-3-030-34175-6_1}
Annabelle McIver and Carroll Morgan.
\newblock Proving that programs are differentially private.
\newblock In Anthony~Widjaja Lin, editor, {\em Programming Languages and
  Systems}, pages 3--18, Cham, 2019. Springer International Publishing.

\bibitem{McIver:2014aa}
Annabelle McIver, Carroll Morgan, Geoffrey Smith, Barbara Espinoza, and Larissa
  Meinicke.
\newblock Abstract channels and their robust information-leakage ordering.
\newblock In Mart{\'{\i}}n Abadi and Steve Kremer, editors, {\em Principles of
  Security and Trust - Third International Conference, {POST} 2014, Held as
  Part of the European Joint Conferences on Theory and Practice of Software,
  {ETAPS} 2014, Grenoble, France, April 5-13, 2014, Proceedings}, volume 8414
  of {\em Lecture Notes in Computer Science}, pages 83--102. Springer, 2014.

\bibitem{8049725}
Ilya Mironov.
\newblock { Rényi Differential Privacy }.
\newblock In {\em 2017 IEEE 30th Computer Security Foundations Symposium
  (CSF)}, pages 263--275. IEEE Computer Society, 2017.

\bibitem{10.5555/3620237.3620329}
Milad Nasr, Jamie Hayes, Thomas Steinke, Borja Balle, Florian Tram\`{e}r,
  Matthew Jagielski, Nicholas Carlini, and Andreas Terzis.
\newblock Tight auditing of differentially private machine learning.
\newblock In {\em Proceedings of the 32nd USENIX Conference on Security
  Symposium}, SEC '23, USA, 2023. USENIX Association.

\bibitem{Smith:2009aa}
Geoffrey Smith.
\newblock On the foundations of quantitative information flow.
\newblock In Luca de~Alfaro, editor, {\em Proc.\ 12th International Conference
  on Foundations of Software Science and Computational Structures (FoSSaCS
  '09)}, volume 5504 of {\em Lecture Notes in Computer Science}, pages
  288--302, 2009.

\bibitem{steinke2022compositiondifferentialprivacy}
Thomas Steinke.
\newblock Composition of differential privacy and  privacy amplification by
  subsampling, 2022.

\bibitem{annurev:/content/journals/10.1146/annurev-statistics-112723-034158}
Weijie~J. Su.
\newblock A statistical viewpoint on differential privacy: Hypothesis testing,
  representation, and blackwell and apos;s theorem.
\newblock {\em Annual Review of Statistics and Its Application}, 12(Volume 12,
  2025):157--175, 2025.

\bibitem{10.1142/S0218488502001648}
Latanya Sweeney.
\newblock k-anonymity: a model for protecting privacy.
\newblock {\em Int. J. Uncertain. Fuzziness Knowl.-Based Syst.},
  10(5):557–570, October 2002.

\bibitem{pmlr-v89-wang19b}
Yu-Xiang Wang, Borja Balle, and Shiva~Prasad Kasiviswanathan.
\newblock Subsampled renyi differential privacy and analytical moments
  accountant.
\newblock In Kamalika Chaudhuri and Masashi Sugiyama, editors, {\em Proceedings
  of the Twenty-Second International Conference on Artificial Intelligence and
  Statistics}, volume~89 of {\em Proceedings of Machine Learning Research},
  pages 1226--1235. PMLR, 16--18 Apr 2019.

\bibitem{10.1145/3093333.3009884}
Danfeng Zhang and Daniel Kifer.
\newblock Lightdp: towards automating differential privacy proofs.
\newblock {\em SIGPLAN Not.}, 52(1):888–901, January 2017.

\end{thebibliography}
